\documentclass[aps,prx,longbibliography,twocolumn,superscriptaddress,floatfix]{revtex4-2}
\usepackage{amstext}
\usepackage{amssymb}
\usepackage{amsmath}
\usepackage{amsthm}
\usepackage{algorithm}
\usepackage{algpseudocode}
\usepackage{hhline}
\usepackage{dsfont}
\usepackage{float}
\usepackage{tikz}
\usepackage{makecell}
\usepackage{siunitx}
\usepackage{bbold}
\usepackage[braket,qm]{qcircuit}
\usepackage{braket}
\usepackage{mathtools}
\usepackage{xcolor}
\usepackage{hyperref}
\usepackage{enumitem}
\hypersetup{colorlinks=true,allcolors=[rgb]{0.1,0.1,0.5}}

\usepackage[utf8]{inputenc}

\usepackage{cleveref}
\Crefname{equation}{Eq.}{Eqs.}
\Crefname{figure}{Fig.}{Figs.}
\Crefname{tabular}{Tab.}{Tabs.}

\newtheorem{theorem}{Theorem}
\newtheorem{lemma}[theorem]{Lemma}
\theoremstyle{remark}
\newtheorem{remark}{Remark}

\definecolor{mygreen}{rgb}{0,0.5,0}

\newcommand{\ugood}{\overline{\text{good}}(c)}
\newcommand{\ubad}{\overline{\text{bad}}(c)}

\begin{document}

\title{Resilience of quantum random access memory to generic noise}

\author{Connor T. Hann}
\affiliation{Departments of Applied Physics and Physics, Yale University, New Haven, Connecticut 06520, USA}
\affiliation{Yale Quantum Institute, New Haven, CT 06520, USA}

\author{Gideon Lee}
\affiliation{Departments of Applied Physics and Physics, Yale University, New Haven, Connecticut 06520, USA}
\affiliation{Yale Quantum Institute, New Haven, CT 06520, USA}
\affiliation{Yale-NUS College, 16 College Avenue West, 138527, Singapore}

\author{S. M. Girvin}
\affiliation{Departments of Applied Physics and Physics, Yale University, New Haven, Connecticut 06520, USA}
\affiliation{Yale Quantum Institute, New Haven, CT 06520, USA}

\author{Liang Jiang}
\affiliation{Departments of Applied Physics and Physics, Yale University, New Haven, Connecticut 06520, USA}
\affiliation{Yale Quantum Institute, New Haven, CT 06520, USA}
\affiliation{Pritzker School of Molecular Engineering, The University of Chicago, Chicago, Illinois 60637, USA}

\begin{abstract}
Quantum random access memory (QRAM)---memory which stores classical data but allows queries to be performed in superposition---is required for the implementation of numerous quantum algorithms. While naive implementations of QRAM are highly susceptible to decoherence and hence not scalable, it has been argued that the bucket brigade QRAM architecture [Giovannetti \textit{et al}., \href{https://journals.aps.org/prl/abstract/10.1103/PhysRevLett.100.160501}{Phys.~Rev.~Lett.~100 160501 (2008)}] is highly resilient to noise, with the infidelity of a query scaling only logarithmically with the memory size. In prior analyses, however, this favorable scaling followed directly from the use of contrived noise models, thus leaving open the question of whether experimental implementations would actually enjoy the purported scaling advantage. In this work, we study the effects of decoherence on QRAM in full generality. Our main result is a proof that this favorable infidelity scaling holds for \emph{arbitrary} error channels (including, e.g., depolarizing noise and coherent errors). Our proof identifies the origin of this noise resilience as the limited entanglement among the memory's components, and it also reveals that significant architectural simplifications can be made while preserving the noise resilience. We verify these results numerically using a novel classical algorithm for the efficient simulation of noisy QRAM circuits. Our findings indicate that QRAM can be implemented with existing hardware in realistically noisy devices, and that high-fidelity queries are possible without quantum error correction. Furthermore, we also prove that the benefits of the bucket-brigade architecture persist when quantum error correction is used, in which case the scheme offers improved hardware efficiency and resilience to logical errors.
\end{abstract}
\maketitle

%%%%%%%%%%%%%%%%%%%%%%%%%%%%%%%%%%%%%%%%%%%%%%%%%%
%%%%%%%%%%%%%%%%%%%%%%%%%%%%%%%%%%%%%%%%%%%%%%%%%%
%%%%%%%%%%%%%%%%%%%%%%%%%%%%%%%%%%%%%%%%%%%%%%%%%%

\section{Introduction}

Numerous quantum algorithms have been proposed that claim speedups over their classical counterparts. Such algorithms typically require that classical data---constituting a classical description of the problem instance---be made available to a quantum processor. Frequently, theoretical constructions called oracles (or black boxes) are invoked to provide this access~\cite{ambainis2004}.  For example, oracles can provide quantum access to classical descriptions of Hamiltonians in quantum simulation algorithms~\cite{berry2012,berry2015,berry2015b,babbush2018,low2019,bauer2020b}, and they are used to encode classical datasets into quantum states in quantum machine learning algorithms~\cite{lloyd2013,wittek2014,adcock2015,biamonte2017,ciliberto2018}. In practice, however, providing quantum access to classical data can be nontrivial, and in order to claim a genuine quantum speedup it is crucial that the details of how such oracles are implemented be specified~\cite{aaronson2015}. 

Quantum random access memory (QRAM)~\cite{giovannetti2008,giovannetti2008a,hong2012,arunachalam2015,dimatteo2020,paler2020a} is a general-purpose architecture for the implementation of quantum oracles. 
QRAM can be understood as a generalization of classical RAM; the classical addressing scheme in the latter is replaced by a quantum addressing scheme in the former.
More precisely, in the case of classical RAM, an address $i$ is provided as input, and the RAM returns the memory element $x_i$ stored at that address. Analogously, in the case of QRAM, a quantum superposition of different addresses $\ket{\psi_\mathrm{in}}$ is provided as input, and the QRAM returns an entangled state $\ket{\psi_\mathrm{out}}$ where each address is correlated with the corresponding memory\- element\-,
\begin{align}
\ket{\psi_\text{in}}&=\sum_{i=0}^{N-1}\alpha_i \ket{i}^A\ket{0}^B \nonumber \\
\xrightarrow[]{\text{QRAM}} \ket{\psi_\text{out}}&=\sum_{i=0}^{N-1}\alpha_i \ket{i}^A\ket{ x_i}^B,
\label{eq:QRAM_def}
\end{align}
where $N$ is the size of the memory~\footnote{Alternatively, one could denote the size of the memory by $2^n$, where $n\equiv \log_2 N$ is the number of qubits in the address register $A$. Our main result, the polylogarithmic (in $N$) scaling of the QRAM query infidelity, could then be equivalently stated as a polynomial scaling (in $n$).}, and the superscripts $A$ and $B$ respectively denote the input and output qubit registers. (In this work, we restrict our attention to the case where the memory elements are classical, though in principle QRAM can also be used to query quantum data.) 
Remarkably, QRAM can perform operation~(\ref{eq:QRAM_def}) in only $O(\log N)$ time, albeit at the cost of $O(N)$ ancillary qubits. The short query time, together with the generality of operation~(\ref{eq:QRAM_def}), makes QRAM appealing for use in many quantum algorithms, especially those that require $O(\log N)$ query times in order to claim exponential speedups. 
Further, QRAM can serve as an oracle implementation in quantum algorithms for machine learning~\cite{lloyd2013,wittek2014,wiebe2014a,adcock2015,biamonte2017,ciliberto2018,kerenidis2016,kerenidis2020}, chemistry~\cite{cao2019a,bauer2020b}, and a host of other areas~\cite{grover1996,schutzhold2003,childs2003,schaller2006,giovannetti2008b,harrow2009,wiebe2012a,lloyd2016,low2018}. 

The idea of QRAM has faced skepticism, however, and the question of whether QRAM can be used to facilitate quantum speedups, either in principle or in practice, has not been definitively settled (see, e.g., Refs.~\cite{aaronson2015,steiger2016}, or the excellent summary in Ref.~\cite{ciliberto2018}). A central practical concern is the seemingly high susceptibility of QRAM to decoherence~\cite{giovannetti2008,arunachalam2015}. As we discuss below, naive implementations of QRAM perform operation (\ref{eq:QRAM_def}) with an infidelity that scales linearly with the size of the memory. Such implementations are not scalable. As the memory size increases, the infidelity grows rapidly without quantum error correction, yet the overhead associated with error correction can quickly become prohibitive because all $O(N)$ ancillary qubits need to be corrected~\cite{dimatteo2020}. 

Refs.~\cite{giovannetti2008,giovannetti2008a} proposed the so-called “bucket-brigade” QRAM architecture as a potential solution to this decoherence problem, though this solution has also faced skepticism. Proponents argue that the bucket-brigade QRAM is highly resilient to noise, in that it can perform operation (\ref{eq:QRAM_def}) with an infidelity that scales only polylogarithmically with the size of the memory. This favorable scaling could allow for high-fidelity queries of large memories without the need for quantum error correction, thereby mitigating the aforementioned scalability problem.
This noise resilience, however, has only been derived for contrived noise models that place severe constraints on the quantum hardware~\cite{giovannetti2008,giovannetti2008a,arunachalam2015}, thus casting doubt on the viability of the bucket-brigade architecture.
Indeed, 
while several proposals for experimental implementations of QRAM have been put forth~\cite{giovannetti2008a,hann2019,kyaw2015,cadellans2015,hong2012}, to our knowledge there has yet to be an experimental demonstration of even a small-scale QRAM~\footnote{Note that the ``random access quantum memories'' demonstrated in Refs.~\cite{naik2017,jiang2019,langenfeld2020} are distinct from QRAM; these experiments do not demonstrate the quantum addressing needed to perform operation~(\ref{eq:QRAM_def}).}.
Absent from this debate has been a fully general and rigorous analysis of how decoherence affects the bucket-brigade architecture. 

In this work, we study the effects of generic noise on the bucket-brigade QRAM architecture. Our main result is that the architecture is far more resilient to noise than was previously thought (our main scaling results are summarized in \Cref{tab:infid_scalings}). We rigorously prove that the infidelity scales only \emph{polylogarithmically} with the memory size even when all components are subject to arbitrary noise channels, and we verify this scaling numerically. 
Remarkably (and perhaps counter-intuitively), this scaling holds even for noise channels where the expected number of errors scales \emph{linearly} with the memory size. 
Our analysis reveals that this remarkable noise resilience is a consequence of the limited entanglement among the memory's components. We leverage this result to show that significant architectural simplifications can be made to the bucket-brigade QRAM, and that so-called ``hybrid'' architectures~\cite{low2018,berry2019,dimatteo2020,paler2020a}, which implement~(\ref{eq:QRAM_def}) with fewer  qubits  but  longer  query  times,  can also be made partially noise resilient. We also show that these benefits persist when quantum error correction is used.
Importantly, the present work shows that a noise-resilient QRAM can be constructed from realistically noisy devices, paving the way for small-scale, near-term experimental demonstrations of QRAM.

\begin{table}[t]
    \centering
    \def\arraystretch{1.5}
    \begin{tabular}{ |p{5.2cm}|c|  }
\hline
\textbf{Architecture} & \textbf{Infidelity scaling}  \\
\hhline{|=|=|}
Fanout QRAM (Sec.~\ref{sec:2}) & $ N \log N$  \\
\hline
Standard BB QRAM (Secs.~\ref{sec:2},\ref{sec:3})  & $ \log^2 N$  \\
\hline
Two-level BB QRAM (Sec.~\ref{sec:5}) & $ \log^3 N$  \\
\hhline{|=|=|}
Hybrid fanout (Sec.~\ref{sec:hybrid}) & $ N \log N + M \log^2 N$  \\
\hline
Hybrid BB (Sec.~\ref{sec:hybrid}) & $ M \log^2 N$  \\
\hline
\end{tabular}
    \caption{Infidelity scalings of QRAM architectures. $N$ denotes the size of the classical memory being queried, and bucket-brigade is abbreviated as BB. The first three architectures have circuit depth $O(\log N)$ and require $O(N)$ qubits. For the hybrid architectures, $M\leq N$ is a tunable parameter that determines the circuit depth, $O(M\log N)$, and the number of qubits, $O(N/M+\log N)$.}
    \label{tab:infid_scalings}
\end{table}

This paper is organized as follows. In \Cref{sec:2}, we give a detailed review of QRAM architectures, and an intuitive explanation for the noise-resilience of the bucket-brigade scheme is provided. { Our main result is presented in~\Cref{sec:3}: we prove that the query infidelity of the bucket-brigade architecture scales only polylogarithmically when its components are subject to generic mixed-unitary error channels (the full proof for arbitrary error channels is given in~\Cref{appendix:proof}). Importantly, these proofs assume that \emph{all} components of the QRAM (both active and inactive) are susceptible to decoherence, in contrast to prior works. The remaining sections provide corollaries, generalizations, and numerical demonstrations of this main result. }
In \Cref{sec:4}, we propose and implement an efficient classical algorithm for the simulation of noisy QRAM circuits, and we use this algorithm to confirm that the bucket-brigade QRAM is resilient to realistic errors. Next, in \Cref{sec:5}, we show that the use of three-level memory elements in the original bucket-brigade architecture is superfluous and that the architecture can be significantly simplified (while maintaining noise resilience) by instead using two-level memory elements. 
In \Cref{sec:hybrid}, we show that the bucket-brigade architecture can also be employed to imbue hybrid architectures with partial noise resilience.
In \Cref{sec:6}, we prove that error-corrected implementations of the bucket-brigade architecture are resilient to logical errors, and we discuss the practical utility of error-corrected QRAM. Finally, in \Cref{sec:7} we conclude by discussing potential applications. 

%%%%%%%%%%%%%%%%%%%%%%%%%%%%%%%%%%%%%%%%%%%%%%%%%%
%%%%%%%%%%%%%%%%%%%%%%%%%%%%%%%%%%%%%%%%%%%%%%%%%%
%%%%%%%%%%%%%%%%%%%%%%%%%%%%%%%%%%%%%%%%%%%%%%%%%%

\section{Quantum random access memory}
\label{sec:2}

In both classical and quantum random accesses memories, each location in memory is indexed by a unique binary address. To read from the memory, an address is provided as input, and the memory element located at that address is returned at the output. In the classical case, transistors are the physical building blocks of the addressing scheme: they act as classical routers, directing electrical signals to the memory location specified by the address bits. Analogously, in the quantum case, quantum routers are the fundamental building blocks of the addressing scheme.  As shown in Fig.~\ref{fig1}(a), a quantum router is a device that directs incident signals along different paths in coherent superposition, conditioned on the state of a routing qubit. For example, if the routing qubit is in state $\ket{0}$ ($\ket{1}$), then a qubit incident on the router is routed to the left (right). If the routing qubit is in a superposition, then the incident qubit is routed in both directions in superposition, becoming entangled with the routing qubit in the process.  Quantum routers can also be understood through the language of quantum circuits [Fig.~\ref{fig1}(b)]; the routing operation is a unitary that can be implemented via a sequence of controlled-SWAP gates (Fredkin gates).  

In this section, we review two QRAM architectures based on quantum routers: the fanout architecture~\cite{nielsen2000} and the bucket-brigade architecture~\cite{giovannetti2008,giovannetti2008a}. 
The fanout architecture is highly susceptible to noise, and we discuss it in order to illustrate how the lack of noise-resilience fundamentally limits QRAM scalability. The noise-resilience of the bucket-brigade architecture is the main focus of this work. 

\subsection{Fanout QRAM architecture}

\begin{figure*}
\centering{}\includegraphics[width=2.0\columnwidth]{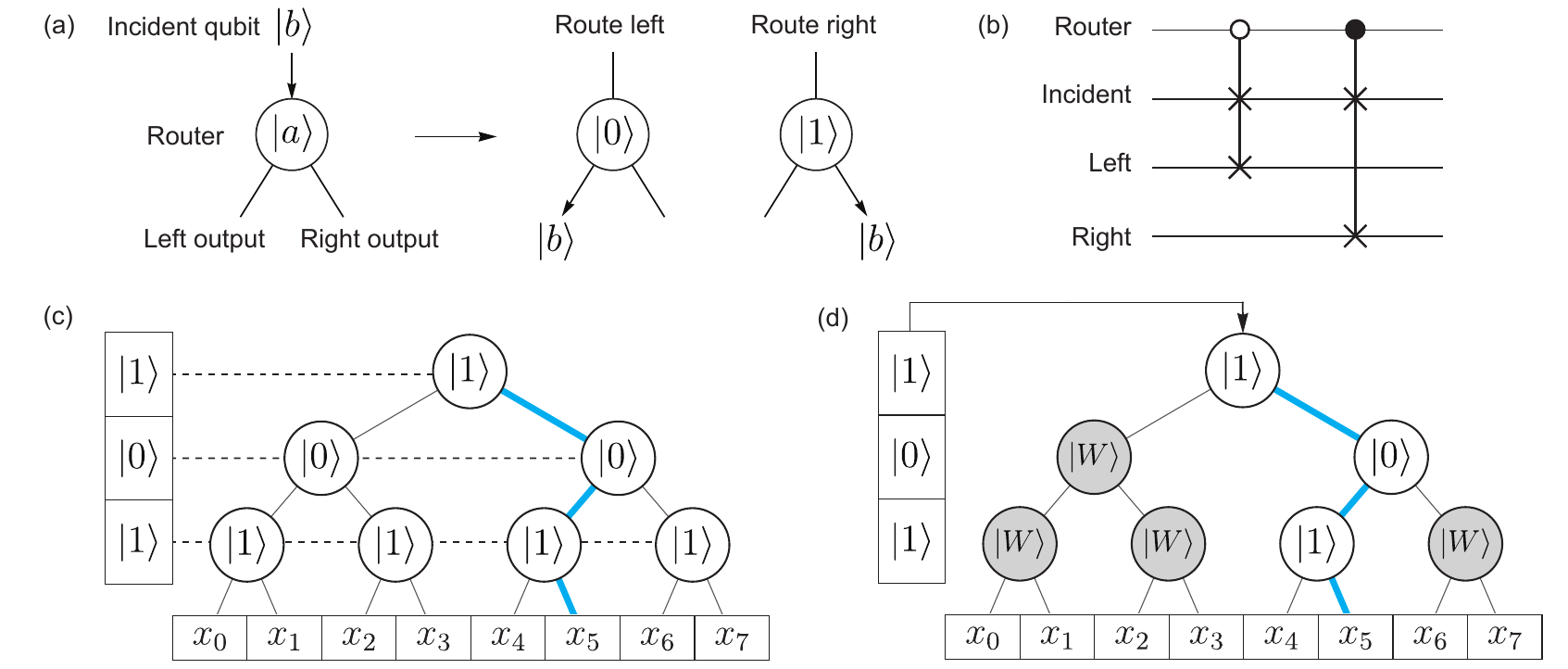}\caption{\label{fig1} QRAM implementations. (a) Quantum router. The router directs an incident qubit $\ket{b}$ at its top port out of either the left or right output ports conditioned on the state $\ket{a}$ of the router. When $\ket{a}=\ket{0}(\ket{1})$, the incident qubit leaves out of the left (right) port. (b) Example of a quantum circuit that implements the routing operation using two controlled-SWAP gates, one conditioned on the control being $\ket{0}$ (open circle) and the other conditioned on the control being $\ket{1}$ (filled circle). (c) Fanout QRAM. Each address qubit controls the states of all routers within the corresponding level of the binary tree. A bus qubit injected at the top node then follows the path (blue) to the specified memory element.  (d) Bucket-brigade QRAM, utilizing routers with three sates: wait $\ket{W}$, route left $\ket{0}$, and route right $\ket{1}$. The address qubits themselves are routed into the tree, carving out a path to the memory.
}
\end{figure*}

A QRAM can be constructed out of quantum routers as shown in Fig 1(c) (see Ref.~\cite{nielsen2000}, Chapter 6). A collection of routers is arranged in a binary tree, with the outputs of routers at one level of the tree acting as inputs to the routers at the next level down. The memory is located at the bottom of the tree, with each of the $N$ memory cells connected to a router at the bottom level. To query the memory, all routing qubits are initialized in $\ket{0}$, and a register of $\log N$ address qubits is prepared in the desired state (in this work, all logarithms are base 2). All routing qubits at level $\ell$ of the tree are then flipped from $\ket{0}$ to $\ket{1}$ conditioned on the $\ell$-th address qubit. To retrieve the memory contents, a so-called bus qubit is prepared in the state $\ket{0}$ and injected into the tree at the top node. The bus follows the path indicated by the routers down to the memory. Upon reaching a memory cell, the contents of that memory cell are copied into the state of the bus (see~\Cref{appendix:copying_data} for details). Note that because we consider \emph{classical} data, the data can be copied without violating the no-cloning theorem. For simplicity, we assume that each memory element $x_i$ is a single bit, 
in which case a single bus qubit suffices to store the memory element (higher-dimensional data can be retrieved using multiple bus qubits).
Finally, the bus is routed back out of the tree via the same path, and all routers are flipped back to $\ket{0}$ in order to disentangle them from the rest of the system.

Importantly, because the routers operate coherently, the above procedure allows one to query multiple memory elements in superposition, as in~(\ref{eq:QRAM_def}). If the address qubits are prepared in a superposition of different computational basis states, the bus is routed to a superposition of different memory locations.

In this architecture, the total time required to perform a query (or, equivalently, the circuit depth) is only $O(\log N)$. The ability to perform queries in logarithmic time can be crucial for algorithms that invoke QRAM in order to claim exponential speedups over their classical counterparts. However, this speed comes at the price of a high hardware cost. To perform operation~(\ref{eq:QRAM_def}), both the fanout and bucket-bridgade architectures require $O(N)$ ancillary qubits to serve as routers.
We discuss practical concerns associated with the high hardware cost more thoroughly in~\Cref{sec:6,sec:7}. 
For context, we note that there is a space-time trade-off in implementing operation~(\ref{eq:QRAM_def}). At the other extreme, there are circuits which implement~(\ref{eq:QRAM_def}) in $O(N)$ time using only $O(\log N)$ qubits~\cite{babbush2018,park2019,deveras2020}, and several implementations that leverage this trade-off have also been proposed~\cite{low2018,berry2019,dimatteo2020,paler2020a}. We discuss the effect of decoherence on these implementations in detail in \Cref{sec:hybrid}.

The fanout architecture is impractical due to its high susceptibility to decoherence.  Each address qubit is maximally entangled with all routers at the respective level of the tree (similar to a GHZ state), so the decoherence of any individual router is liable to ruin the query. As an example, suppose that the routers are subject to amplitude damping errors. The loss of an excitation from any router at level $\ell$ collapses all other level-$\ell$ routers and the $\ell$-th address qubit to the $\ket{1}$ state. Any terms in the superposition where the $\ell$-th address qubit was in the $\ket{0}$ state prior to the error are thus projected out, thereby reducing the fidelity by a factor of 2 on average. 

More generally, suppose that each router suffers an error with probability $\varepsilon$ at each time step during the query. The final state $\Omega$ of the full system (address, bus, and routers) can then be written as a statistical mixture 
\begin{equation}
\Omega = (1-\varepsilon)^{T(N-1)} \;\Omega_{\mathrm{ideal}} + \ldots
\label{eq:fanout_mixture}
\end{equation}
where $\Omega_\mathrm{ideal}$ is the error-free state, $T=O(\log N)$ is the number of time steps required to perform a query, and ``$\ldots$'' denotes all states in the mixture where at least one of the $N-1$ routers has suffered an error. We define the \emph{query fidelity} as
\begin{equation}
F = \braket{\psi_\text{out}| \mathrm{Tr}_R\left(\Omega \right) |\psi_\text{out}},
\end{equation}
where $\mathrm{Tr}_R$ indicates the partial trace over the routers. The routers are traced out because only the address and bus registers are passed on to whatever algorithm has queried the QRAM; the routers are ancillae whose only purpose is to facilitate the implementation of operation~(\ref{eq:QRAM_def}). 

As illustrated by the amplitude-damping example, the problem with the fanout implementation is that the no-error state $\Omega_\mathrm{ideal}$ is generally the only state in the mixture~(\ref{eq:fanout_mixture}) with high fidelity.  Neglecting the low-fidelity states, the query infidelity scales as
\begin{align}
1-F &\sim \varepsilon N T,
\end{align}
to leading order in $\varepsilon$. We refer to this linear scaling of the infidelity with the memory size as \emph{unfavorable} because error probabilities $\varepsilon \ll 1/NT$ are required to perform queries with near-unit fidelity.  This stringent requirement severely constrains the size of fanout QRAMs. For example, error probabilities $\varepsilon \sim 10^{-3}$ would restrict the maximum size of a high-fidelity fanout QRAM to less than $N\sim 100$ memory cells. While quantum error correction can be used to suppress the error rates in principle, the additional hardware overhead can be prohibitive~\cite{dimatteo2020} because all $O(N)$ routers must be error corrected (see \Cref{sec:6} for a more detailed discussion of error correction and the associated overhead). 
Thus, because of its high susceptibility to decoherence, the fanout architecture is not regarded as scalable. 

\subsection{Bucket-brigade QRAM architecture}

In Ref.~\cite{giovannetti2008}, two modifications to the fanout architecture were proposed, and it was argued that the modified architecture, 
termed the ``bucket-brigade'' [Fig.~\ref{fig1}(d)], is highly resilient to noise. The first modification is that the two-level routing qubits are replaced with three-level routing qutrits. In addition to the $\ket{0}$ (route left) and $\ket{1}$ (route right) states, each router also has a third state, $\ket{W}$ (wait). We refer to the states $\ket{0},\ket{1}$ as \emph{active}, and the state $\ket{W}$ as \emph{inactive}. 
We assume that all routers are initialized in the $\ket{W}$ state, and that the action of the routing operation [Fig.~\ref{fig1}(b)] is trivial when the routing qutrit is in the $\ket{W}$ state. 
(In \Cref{sec:5}, we show that these assumptions may be relaxed, but we make them here for concreteness.)  
Each router's incident and output modes are also now taken to be physical three-level systems, 
and each address qubit is encoded within a two-level subspace of a physical three-level system.

{
We have borrowed the terminology of \emph{active} and \emph{inactive} routers from earlier works on QRAM~\cite{giovannetti2008,arunachalam2015}. 
While these prior works assumed inactive routers to be free from decoherence, we stress that we make no such assumption here. 
As is discussed further in \Cref{subsec:2c}, we assume that routers are prone to decoherence regardless of whether they are active or inactive. 
In this work, we define the terms ``active'' and ``inactive'' only as labels for the different subspaces of a router's Hilbert space: active $\equiv$ span($\ket{0},\ket{1}$) and inactive $\equiv$ span($\ket{W}$). We do not assume that either of these subspaces has any special properties (e.g.,~different decoherence rates).   
}

The second modification is that the address qubits are themselves routed into the tree during a query. When an address qubit encounters a router in the $\ket{0}$ ($\ket{1}$) state, it is routed to the left (right) as usual. 
When an address qubit encounters a router in the $\ket{W}$ state, the states of the router and incident mode are swapped, 
so that the router's state becomes $\ket{0}$ ($\ket{1}$) when the incident address was $\ket{0}$ ($\ket{1}$). 
The physical implementation described in Ref.~\cite{giovannetti2008} provides a helpful example to visualize how these operations could be realized: the authors envisage the routers as three-level atomic systems, with the address qubits encoded in the polarization states of flying photons. (Note that the two polarization states constitute a two-level subspace of a physical three-level system, since the photonic mode may also be in the vacuum state.)  
When a photon encounters an atom in the $\ket{W}$ state, it is absorbed, and in the process it excites the atom to the $\ket{0}$ or $\ket{1}$ state conditioned on its polarization. When subsequent photons encounter the excited atom, they are routed accordingly. 
These operations can also be described using the conventional quantum circuit model, and in \Cref{appendix:circuit} we provide a full circuit diagram. 

To query the memory, the address qubits are sequentially injected into the tree at the root node. 
The first address qubit is absorbed by the router at the root node, exciting it from $\ket{W}$ to the $\{\ket{0},\ket{1}\}$ subspace in the process. 
The second address qubit is routed left or right, conditioned on the state of the router at the root node. 
The state of the first address qubit thereby dictates the routing of the second.
The second address is subsequently absorbed by one of the routers at the second level of the tree. 
The process is repeated, with the earlier addresses controlling the routing of later ones, carving out a path of active routers from the root node to the specified memory element. 
Once all address qubits have been routed into the tree, the bus qubit is routed down to the memory and the data is copied as before (see \Cref{appendix:copying_data}). Finally, the bus and all address qubits are routed back out of the tree in reverse order to disentangle the routers. Here again, we emphasize that multiple memory elements can be queried in superposition, as in~(\ref{eq:QRAM_def}), because all routing operations are performed coherently.

\subsection{Noise resilience: overview and conceptual explanation}
\label{subsec:2c}

\begin{figure*}[ht]
\centering{}\includegraphics[width=2.0\columnwidth]{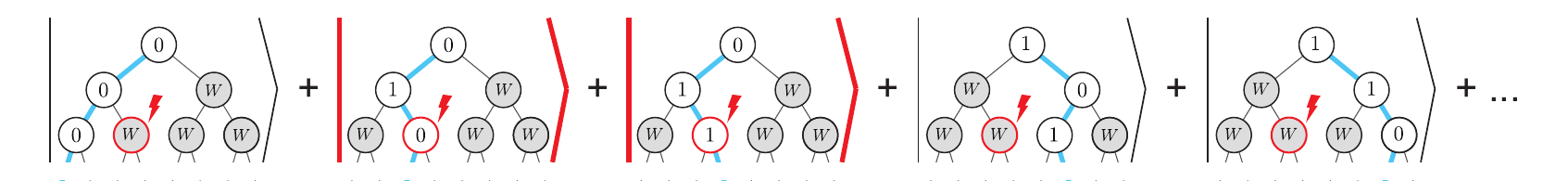}\caption{\label{fig1pt5} Conceptual picture of noise resilience. Each ket represents the state of the QRAM when a different memory element is queried, with the superposition of kets representing a superposition of queries to different elements. When a router $r$ suffers an error (red lightning bolt), it corrupts only the subset of queries where $r$ is active (indicated by thick red kets); other queries in the superposition succeed regardless. Because most routers are only active in a small fraction of queries, most queries succeed and the total infidelity is low.
}
\end{figure*}

The bucket-brigade architecture is clearly resilient to certain types of noise. For example, Ref.~\cite{arunachalam2015} studied the bucket-brigade QRAM with routers subject to $\ket{0} \leftrightarrow \ket{1}$ bit-flip errors, with the $\ket{W}$ states assumed to be error free. In this case, the expected number of errors is only $\varepsilon \log N$, because only the $\log N$ active routers are prone to errors~\footnote{This is true even when all memory elements are queried in superposition; each router is in a superposition of active and inactive states, but there are only $\log N$ active routers in each branch of the superposition corresponding to a definite address.}. The expected number of errors also scales with $\log N$ for the error model considered in Refs.~\cite{giovannetti2008,giovannetti2008a,hong2012}, where gates involving inactive routers are assumed to be error free. 

For these error models, the query infidelity is
\begin{equation}
1-F\sim \varepsilon T \log^\alpha N.
\end{equation}
to leading order in $\varepsilon$, where $\alpha$ is some constant, and we recall that $T = O(\log N)$ is the number of time steps.  We refer to this logarithmic scaling of the infidelity with the memory size as \emph{favorable} because queries can be performed with near-unit fidelity so long as the error rate satisfies $\varepsilon \ll 1/T\log^\alpha N$. This is a much more forgiving requirement; memories of exponentially larger size can be queried relative to the fanout architecture. Indeed, the exponential improvement in scalability suggests that quantum error correction is not required to query large memories with high fidelity, provided physical error rates are sufficiently low. 

Unfortunately, the above error models can be poor approximations of the noise in actual quantum hardware. In these contrived models, inactive routers are assumed to be completely free from decoherence. More realistically, all routers will be prone to decoherence, independent of whether they are active or inactive. For example, though several proposals for experimental implementations of the bucket-brigade scheme have been put forth~\cite{giovannetti2008a,hong2012,kyaw2015,hann2019,cadellans2015}, none have proposed a method of engineering routers that are free from decoherence when inactive.
While one can conceive of implementations in which inactive routers have decoherence rates which are nonzero but far smaller than those of active routers, it is not obvious whether such implementations would enjoy the favorable infidelity scaling. 
Indeed, Ref.~\cite{giovannetti2008} conjectured that decoherence of inactive routers could significantly increase the infidelity in this case, owing to the exponentially larger number of inactive routers.
Furthermore, Refs.~\cite{giovannetti2008,arunachalam2015} portray the favorable infidelity scaling as a direct consequence of the assumption that inactive routers are decoherence-free.

It is thus natural to ask whether the favorable scaling still holds when inactive routers are not assumed to be decoherence-free.
Relaxing this assumption causes the expected number of errors to increase \emph{exponentially}, from $O(\log N)$ to $O(N)$. 
Because the expected number of errors in the fanout architecture is also $O(N)$, one might naively expect that the favorable infidelity scaling no longer holds. 
However, in the next section we prove that this is not the case. Perhaps surprisingly, the infidelity of the bucket-brigade architecture still scales favorably despite the exponential increase in the expected number of errors. Moreover, the favorable scaling holds for \emph{arbitrary} error channels.

The noise-resilience of the bucket-brigade architecture can be understood intuitively as a consequence of the minimal entanglement among the routers, see \Cref{fig1pt5}. Suppose one queries all memory locations in equal superposition. Then in both the fanout and bucket-brigade architectures, all of the routers are entangled. However, the degree to which each router is entangled with the rest of the system is quite different between the two architectures. This difference can be quantified by computing the entanglement entropy for a given router
\begin{equation}
    S(\rho) = -\mathrm{Tr}\left[\rho \log \rho \right]
\end{equation}
where $\rho$ is the reduced density matrix of the router, obtained by tracing out the rest of the system. In the fanout architecture, each router is maximally entangled with the rest of the system; the reduced density matrix is the maximally mixed state $\rho = I/2$ (recall the fanout architecture employs two-level routers), for which $S(\rho) = 1$. In contrast, in the bucket-brigade architecture, the entanglement entropy of a router depends on its location within the tree.  A router at level $\ell$ (0-indexed) of the tree is only active in $N2^{-\ell}$ of the $N$ different branches of the superposition. As a result, the entanglement entropy decreases exponentially with depth, $S(\rho) \sim 2^{-\ell} $. Routers deeper down in the tree are nearly disentangled from the system, and their decoherence only reduces the query fidelity by an exponentially decreasing amount. 
Thus, despite the fact that exponentially many such errors typically occur, the overall fidelity can remain high. 
More precisely, if we posit that the infidelity associated with an error in a router at level $\ell$ scales as $\sim 2^{-\ell}$ due to the limited entanglement, and that $\varepsilon T\, 2^\ell$ such routers suffer errors on average, then the total infidelity scales as
\begin{equation}
    1-F\sim \sum_{\ell=1}^{\log N} \left(2^{-\ell}\right) \left(\varepsilon T 2^\ell\right) = \varepsilon T \log N.
\end{equation}
The infidelity scales only logarithmically with $N$ because the exponential increase in the expected number of errors with $\ell$ is precisely cancelled by the exponential decrease in the infidelity associated with each.
We rigorously justify these claims in the next section.

%%%%%%%%%%%%%%%%%%%%%%%%%%%%%%%%%%%%%%%%%%%%%%%%%%
%%%%%%%%%%%%%%%%%%%%%%%%%%%%%%%%%%%%%%%%%%%%%%%%%%
%%%%%%%%%%%%%%%%%%%%%%%%%%%%%%%%%%%%%%%%%%%%%%%%%%

\section{Proof of noise resilience}
\label{sec:3}

In this section, we prove that the query infidelity of the bucket brigade architecture is upperbounded by
\begin{equation}
1-F  \leq A\varepsilon T\log N ,
\label{eq:general_bound}
\end{equation}
where $T = O(\log N)$ is the time required to perform a query, $\varepsilon$ is the probability of error per time step, and $A$ is a constant of order $1$. This bound holds even when all $N$ memory elements are queried in superposition, and it holds for arbitrary error channels, including, e.g., depolarizing errors and coherent errors. Moreover, throughout this paper we assume no special structure in the classical data $x_i$, so our bounds hold independent of the data. 

{ Our proof is based on a careful analysis of how errors can propagate throughout the QRAM.} Accordingly, we begin by defining our error model. We suppose that each routing qutrit is subject to an error channel in the form of a generic completely-positive trace-preserving map,
\begin{equation}
\rho\rightarrow \mathcal{E}(\rho) = \sum_{i}K_m \rho K_m^\dagger, 
\label{eq:general_channel}
\end{equation}
where the Kraus operators $K_m$ obey the completeness relation $\sum_m K_m^\dagger K_m = I$. The error channel is applied simultaneously to all routers at discrete time steps throughout the query (see Eq.~(\ref{eq:omega_c}) below). In \Cref{appendix:proof}, we prove that the bound~(\ref{eq:general_bound}) holds for arbitrary error channels of the form~(\ref{eq:general_channel}). For the sake of brevity and simplicity, however, here we restrict our attention to channels where (i) there is a no-error Kraus operator, $K_0$, that is proportional to the identity, and (ii) the remaining Kraus operators are proportional to unitaries, $K_m^\dagger K_m \propto I$. Under these restrictions,
\begin{equation}
\mathcal{E}(\rho) = (1-\varepsilon)\rho + \sum_{m>0} K_m \rho K_m^\dagger,
\label{eq:mixed_unitary_channel}
\end{equation}
for some $\varepsilon \in [0,1]$.
An operational interpretation of this channel is that one of the errors $K_{m>0}$ occurs with probability $\varepsilon$, and no error occurs with probability $1-\varepsilon$. 
Experimentally relevant examples include bit-flip, dephasing, and depolarizing channels. 
The restriction to this form of mixed-unitary channel allows us to make two assumptions that greatly simplify the proof: (i) the probability that an error occurs is independent of the router state, and (ii) the no-error backaction $K_0\propto I$ is trivial. {We make no further assumptions about the Kraus operators, and we stress that they may act non-trivially on the inactive state $\ket{W}$, meaning that inactive routers can decohere.}

It is important to note that this error model only describes decoherence of the routing qutrits; a router's incident and output modes may also decohere, and there may be errors in the gates that implement the routing operation. 
At the end of this section, we prove that the bound~(\ref{eq:general_bound}) still holds when including these other errors, but we neglect them for now to simplify the discussion.

The proof proceeds by direct calculation. To bound the infidelity, we first write the final state $\Omega$ as a sum over different \emph{error configurations},
\begin{equation}
\Omega = \sum_c p(c) \Omega(c),
\label{eq:final_mixture}
\end{equation}
where an error configuration $c$ specifies which Kraus operator is applied to each router at each time step.
Here, $p(c)$ is the probability of configuration $c$, and the pure state $\Omega(c) = \ket{\Omega(c)}\bra{\Omega(c)}$ is the corresponding final state of the system (both quantities are defined more formally below). 
The fidelity is thus given by,
\begin{equation}
\label{eq:fid_sum_over_c}
F = \sum_c p(c) F(c),
\end{equation}
where 
\begin{equation}
F(c) = \braket{\psi_\text{out}|\mathrm{Tr}_R\Omega(c)|\psi_\text{out}}
\label{eq:configuration_fidelity}
\end{equation}
is the query fidelity of the state $\Omega(c)$. Our approach is to place an upper bound on the infidelity by deriving an upper bound on $1-F(c)$.

Let us formally define $\Omega(c)$ and $p(c)$. A QRAM query consists of $O(N)$ routing operations [Fig.~\ref{fig1}(b)] performed in a predetermined sequence.
By design, many of these operations commute and can be performed in parallel, so that the entire operation can be written as a quantum circuit with depth $T=O(\log N)$ (see \Cref{appendix:circuit} for circuit diagram).
More precisely, operation~(\ref{eq:QRAM_def}) can be written as,
\begin{equation}
\ket{\psi_\text{out}}\ket{\mathcal W} = U_T \ldots U_2 U_1 \ket{\psi_\text{in}}\ket{\mathcal W},
\end{equation}
where $\ket{\mathcal W} = \ket{W}^{\otimes (N-1)}$ is the initial state of the routers, and $U_t$ is a constant-depth circuit. 
Now, let $K_{c(r,t)}$ denote the Kraus operator applied to router $r$ at time step $t$, and define the composite Kraus operator $K_{c(t)} = \bigotimes_{r=1}^{N-1}K_{c(r,t)}$ [see Fig.~\ref{fig2pt5}(a)]. 
The final state $\ket{\Omega(c)}$ is
\begin{equation}
\ket{\Omega(c)} =\frac{1}{\sqrt{p(c)}} \left[U_T K_{c(T)} \ldots U_1K_{c(1)}\right] \ket{\psi_\text{in}}\ket{\mathcal{W}},
\label{eq:omega_c}
\end{equation}
The requirement that $\ket{\Omega(c)}$ is normalized defines the probability $p(c)$ of obtaining state $\Omega(c)$ in the  mixture~(\ref{eq:final_mixture}).
Note that $\sum_c p(c) = 1$ follows from the Kraus operators' completeness relation.

\begin{figure}[tb]
\centering{}\includegraphics[width=1.0\columnwidth]{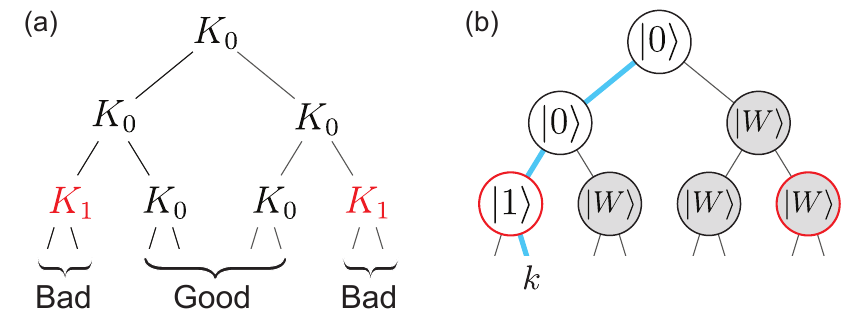}\caption{\label{fig2pt5} Error configurations. (a) Example composite Kraus operator $K_{c(t)}$. The single-router Kraus operators $K_{c(r,t)}$ comprising the tensor product $K_{c(t)}$ are arranged geometrically according to the routers on which they act. Branches of the tree are classified as either good or bad according to the locations of the errors $K_{m>0}$. 
(b) Query to an element $k \not\in g(c)$. Routers are labelled with their ideal, error-free states, and routers outlined in red suffer errors. Because one of the active routers suffers an error, the query is liable to fail. 
}
\end{figure}

For a given error configuration $c$, it is convenient to classify branches of the tree as either \emph{good} or \emph{bad}, depending on whether errors $K_{m>0}$ are ever applied to the routers in the branch [Fig.~\ref{fig2pt5}(a)]. More precisely, let $\mathbf{i}$ denote the set of all routers in the $i$-th branch of the tree (corresponding to address $i$), and let $\mathbf{c}$ denote the set of all routers which have an error $K_{m>0}$ applied to them at some time step. A branch $i$ is defined to be good if $\mathbf{i} \cap \mathbf{c} = \varnothing$, and bad otherwise. To keep the notation simple, we use $g(c)$ to denote set of good branches.
As illustrated in Fig.~\ref{fig2pt5}(b), queries to addresses  $i\not \in g(c)$ are liable to fail because they rely on routers that suffer errors. 

\begin{figure}[htbp]
\centering{}\includegraphics[width=1.0\columnwidth]{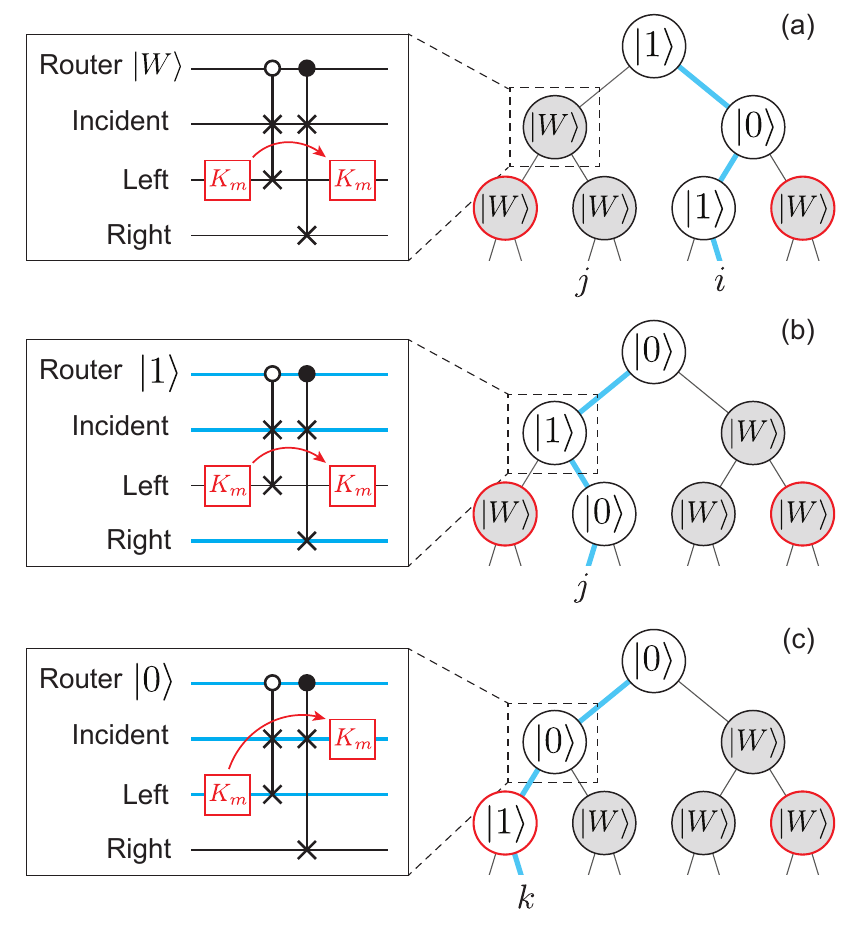}\caption{\label{fig2} Error propagation.
(a,b) Constrained propagation during queries to elements $\in g(c)$. The error in the leftmost router can propagate upward into the left output of the router indicated by the dashed box. The circuits on the left show that the error does not propagate further, regardless of whether the router is inactive (a) or active (b). 
In the circuit diagrams, red boxes denote errors $K_{m>0}$, and the red arrows indicate how the error propagates (i.e.~how the error transforms under conjugation by the routing operation). 
(c) Error propagation is not constrained during queries to elements $\not\in g(c)$.
Note that the state of the router dictates how the error propagates in these examples. 
}
\end{figure}

The main observation underlying our proof is that the propagation of errors is constrained when memory elements $\in g(c)$ are queried. Roughly speaking, errors do not propagate from bad branches into good branches.
More precisely, for any $i,j\in g(c)$, errors do not propagate into branch $j$ during a query to element $i$. We illustrate this fact with two examples, shown in Figs.~\ref{fig2}(a,b).
In general, errors in the bad branches can propagate. They can even propagate into an output mode of a router $r$ in branch $j\in g(c)$, but they can never propagate into branch $j$.
Fig.~\ref{fig2}(a) shows an example of how such an error propagates through $r$'s routing operation in the case where a memory element $i\neq j$ is queried.
Because $j\in g(c)$, $r$'s routing qutrit suffers no errors and is thus in $\ket{W}$. 
The action of the routing operation is trivial for a router in $\ket{W}$, so the error does not propagate to other modes. 
(We reiterate that we are assuming error-free gates; gate errors are discussed at the end of this section.) 
Similarly, Fig.~\ref{fig2}(b) shows an example of how errors propagate in the case where $j$ is queried.   
The error-free routing qutrit is in $\ket{1}$, so the routing operation acts non-trivially on only the incident and right output modes. The error in the left output mode does not propagate upward. For comparison, in \Cref{fig2}(c) we illustrate that the propagation of errors is not constrained in this way when memory elements $k\not\in g(c)$ are queried. As an aside, we note that the constrained error propagation can be understood as a sort of error transparency~\cite{vy2013,kapit2018,ma2020}: when elements $\in$ g(c) are queried, the errors in the bad branches commute with the routing operations in the good branches.

The constrained propagation of errors has two important consequences. The first is that a query to memory element $i\in g(c)$ always succeeds, meaning that the address and bus registers are in the desired state $\ket{i}^A\ket{x_i}^B$ at the end of the query. This follows from the fact that errors cannot propagate to any of the routers in branch $i$. 
The second consequence is that, 
if multiple memory elements $i,j,\ldots \in g(c)$ are queried in superposition, the address and bus registers are disentangled from the routers at the end of the query. 
This follows from the fact that errors are restricted to propagate within the bad branches, and their propagation is unaffected by routers outside these branches. Figs.~\ref{fig2}(a,b) provide an example. As a result, even though errors can propagate non-trivially among the bad branches during the query,
the final state of the routers is independent of which memory element in $ g(c)$ is queried.

It follows that the final state $\ket{\Omega(c)}$ can be written as
\begin{equation}
\ket{\Omega(c)} = \ket{\mathrm{good}(c)} + \ket{\mathrm{bad}(c)},
\label{eq:final_system_state}
\end{equation} 
with 
\begin{equation}
   \ket{\mathrm{good}(c)}=\left(\sum_{i\in g(c)}  \alpha_i \ket{i}^A\ket{x_i}^B\right)\ket{f(c)}^R.
   \label{eq:final_system_state_good}
\end{equation}
Here, $\ket{f(c)}^R$ denotes the final state of the routers with respect to the good branches, and $\ket{\textrm{bad}(c)}$ contains the $i\not\in g(c)$ terms. 
We now use the expression~(\ref{eq:final_system_state}) to place a lower bound on $F(c)$. First notice that
\begin{equation}
    F(c)\geq \left|\braket{\psi_\mathrm{out},f(c)|\Omega(c)}\right|^2,
    \label{eq:basic_bound}
\end{equation}
which can be obtained by performing the partial trace in Eq.~(\ref{eq:configuration_fidelity}) using a basis that contains the state $\ket{f(c)}$ and neglecting the contributions from other states. Then, defining $\Lambda(c)$ as the weighted fraction of good branches, 
\begin{equation}
\Lambda(c) = \braket{\mathrm{good}(c)|\mathrm{good}(c)} = \sum_{i\in g(c)} |\alpha_i|^2
\end{equation}
we have that
\begin{align}
    \braket{\psi_\mathrm{out},f(c)|\mathrm{good}(c)} &= \Lambda(c) \label{eq:good_overlap} \\
    \left|\braket{\psi_\mathrm{out},f(c)|\mathrm{bad}(c)}\right| &\leq 1-\Lambda(c).
    \label{eq:bad_overlap}
\end{align}
To obtain the inequality~(\ref{eq:bad_overlap}) we have used the fact that $\ket{\Omega(c)}$ is normalized and that $\braket{\mathrm{good}(c)|\mathrm{bad}(c)} = 0$. The latter follows from the orthogonality of different initial address states, $\braket{i|j}^A = 0$ for $i\neq j$, and the fact that all subsequent operations, including the Kraus operators, are unitary and thus preserve inner products (this follows from our earlier restriction to mixed-unitary error channels; general channels are covered by the proof in \Cref{appendix:proof}). 
Plugging Eqs.~(\ref{eq:final_system_state}),~(\ref{eq:good_overlap}) and~(\ref{eq:bad_overlap}) into the bound~(\ref{eq:basic_bound}) and applying the reverse triangle inequality allows us to bound the infidelity as a function of $\Lambda(c)$,
\begin{equation}
    F(c) \geq 
    \begin{cases} 
      (2\Lambda(c)-1)^2, & \Lambda(c)\geq 1/2, \\
      0, & \Lambda(c) < 1/2.
      \label{eq:config_fid_bound}
   \end{cases}
\end{equation}

To proceed further, we compute the expected fraction of good branches, $\mathbb{E}(\Lambda)$, where the expectation value is taken with respect to the distribution of error configurations, i.e.~$\mathbb{E}(f) = \sum_c p(c) f(c)$. This expectation value can be computed recursively for trees of increasing depth. Let $\mathbb{E}_d(\Lambda)$ denote the expected fraction of good branches for a depth-$d$ tree. For a depth-1 tree, expected fraction is equivalent to the probability that the lone router never suffers an error, $\mathbb{E}_1(\Lambda)=(1-\varepsilon)^T$.
For deeper trees, the expected fraction of error-free routers at each level is $(1-\varepsilon)^T$, so we have the recursive rule
\begin{equation}
    \mathbb{E}_{d+1}(\Lambda) = (1-\varepsilon)^T\mathbb{E}_{d}(\Lambda).
\end{equation}
Applying this rule to the initial condition $\mathbb{E}_1(\Lambda)$, we obtain 
\begin{equation}
    \mathbb{E}_{\log N}(\Lambda) = (1-\varepsilon)^{T\log N}.
\end{equation}

We can now combine the above results to bound the infidelity. We have that
\begin{align}
    \label{eq:infid_bound_intermediate0}
    F &= \mathbb{E}(F) \geq \mathbb{E}(\sqrt F)^2 \\
    \label{eq:infid_bound_intermediate}
    &\geq \left[ 2\mathbb{E}_{\log N}(\Lambda)-1\right]^2 \\
    \label{eq:infid_bound_intermediate2}
    &= \left[2(1-\varepsilon)^{T\log N}-1\right]^2,
\end{align}
where the second inequality follows from~(\ref{eq:config_fid_bound})  under the assumption that $\mathbb{E}(\Lambda_{\log N})\geq{1/2}$.
Applying Bernoulli's inequality yields the desired result,
\begin{align}
\label{eq:infid_bound}
   1-F \leq  4\varepsilon T \log N,
\end{align}
which holds for 
$\varepsilon {T\log N} \leq 1/4$. 
{
This bound is our main result, and we stress that it holds even when all $N$ elements are queried in superposition, and that it was derived under the assumption that all routers are susceptible to decoherence, regardless of whether they are active or inactive.}
%We stress that this bound holds even when all $N$ elements are queried in superposition. 

We offer two additional remarks on the proof. First, we reiterate that while the above proof holds only for mixed-unitary error channels, in fact the favorable infidelity scaling holds for arbitrary error channels, which we prove in \Cref{appendix:proof}. Second, the favorable scaling can be interpreted as a consequence of the limited entanglement among the routers, as discussed in \Cref{sec:2}. This limited entanglement manifests in \Cref{eq:final_system_state,eq:final_system_state_good}.
The fact that a router at level $\ell$ is active in only $N2^{-\ell}$ of the $N$ branches implies both that the router's entanglement entropy decreases exponentially with $\ell$, and that only $N2^{-\ell}$ branches are corrupted when it suffers an error. 

We conclude this section by describing four simple extensions of the proof that cover other cases of interest:

\textit{1.~Initialization errors.} Suppose that each router has some probability $\varepsilon$ of not being initialized to $\ket{W}$ prior to the query. Such errors can be viewed as router errors of the form (\ref{eq:general_channel}) that occur during the $0$-th time step. As such, they are also covered by the proof provided one replaces $T\rightarrow T+1$ in the equations above. In \Cref{sec:5}, we show that, in fact, one can make an even stronger statement: the infidelity scales favorably even when the QRAM is initialized in an arbitrary state.

\textit{2.~Gate errors.} Faulty implementation of the routing operation can be described without loss of generality as a composition $\mathcal D \circ \mathcal R$, of some error channel $\mathcal D$ followed by the ideal routing operation $\mathcal R$. Provided that $\mathcal D$'s Kraus operators are proportional to unitaries, and that there is a no-error Kraus operator proportional to the identity, then $\mathcal D$ can also be written in the form~(\ref{eq:mixed_unitary_channel}), and the proof proceeds as above. Note that the propagation of errors is still constrained in the case of gate errors because all routing gates in good branches are error-free by construction. 

{
\textit{3.~Alternate gate sets.} We have defined the routing operation as a sequence of two controlled-SWAP gates~[Fig.~1(b)], but this same operation could also be decomposed into other types of gates, e.g.~into Toffolis, or Clifford + T gates. The bound (28) holds for any choice of gate decomposition. 
To see that the bound holds, consider that any error that propagates non-trivially through a given routing operation can be categorized as occurring either $\emph{before}$ or $\emph{during}$ that operation. 
The propagation of errors that occur before the operation is determined solely by the conjugation of the error with the entire routing operation (Fig.~4), which is unaffected by the choice of decomposition. 
In contrast, the propagation of errors that occur during the operation will generally depend on the choice of the decomposition. However, such errors can equivalently be described as a faulty implementation of the routing operation itself, so they do not spoil the favorable error scaling by the argument in the previous paragraph. 

\textit{4.~Correlated errors.} The noise resilience also persists in the presence of correlated errors that afflict a constant number of adjacent routers in the tree. The proof assumes that if any error (correlated or otherwise) occurs in a branch, then that branch does not contribute to the fidelity. 
As such, whether an error afflicts only a given router $r$ or also some of $r$'s child routers lower in the tree is irrelevant to the proof. The effects of correlated errors can thus be incorporated simply by augmenting $\varepsilon$ to also include the probability that a router is among those afflicted by a correlated error. 
For correlated errors afflicting only a constant number of adjacent routers, the resulting increase in $\varepsilon$ is independent of $N$, so the query infidelity still scales only polylogarithmically with $N$. 
}

%%%%%%%%%%%%%%%%%%%%%%%%%%%%%%%%%%%%%%%%%%%%%%%%%%
%%%%%%%%%%%%%%%%%%%%%%%%%%%%%%%%%%%%%%%%%%%%%%%%%%
%%%%%%%%%%%%%%%%%%%%%%%%%%%%%%%%%%%%%%%%%%%%%%%%%%

\section{Classical simulation of noisy QRAM circuits}
\label{sec:4}

In this section, we verify the bound (\ref{eq:infid_bound}) through numerical simulation of noisy QRAM circuits. While full state vector simulations require $\exp(N)$ memory and quickly become intractable as the QRAM size grows, our simulations are enabled by a novel classical algorithm with space and time complexity $\mathrm{poly}(N)$. 

The main observation underlying the algorithm is that any quantum circuit consisting of the following elements can be simulated efficiently classically: state preparation in the computational basis, and gates from the set \{SWAP, controlled-SWAP\}. Such circuits are essentially classical---the system begins in a definite computational basis state, and the SWAP-type gates act only as permutations so that the system remains in a computational basis state through every step of the circuit. The simulation proceeds simply by tracking the (classical) state of the system. 
Furthermore, for initial states that are a superposition of polynomially-many different computational basis states, it follows from linearity that the action of any circuit composed of these SWAP-type gates can also be efficiently simulated.
QRAM circuits can thus be efficiently simulated because they consist of SWAP-type gates acting on $O(N)$ qubits or qutrits, and the system is initialized in a superposition of only $O(N)$ computational basis states (one for each address). 
In fact, QRAM circuits are examples of so-called efficiently computable sparse (ECS) operations, whose efficient classical simulation is described in Ref.~\cite{nest2010}.

For context, we note that this approach is similar in spirit to the Gottesman-Knill theorem~\cite{gottesman1998}, which states that any Clifford circuit with preparation and measurement in the computational basis can be simulated classically in polynomial time. Because QRAM circuits necessarily employ non-Clifford gates (controlled-SWAP), however, the theorem does not directly apply. Still, the similarities are apparent: restricting the allowed gates and state preparations enables an efficient classical description of the system, making efficient simulation possible. 

In addition, for a wide variety of error models, noisy QRAM circuits can be simulated efficiently using Monte Carlo methods. 
To simulate noisy circuits, the space of error configurations is randomly sampled according to the distribution $p(c)$. For each sampled configuration $c$ from a set of samples $S$, we compute the final system state $\ket{\Omega(c)}$, and we obtain the fidelity by averaging $F = \frac{1}{|S|}\sum_{c\in S} F(c)$.
This sampling procedure is efficient provided that two criteria are satisfied: first, that the state $\ket{\Omega(c)}$ is efficiently computable, and second, that sampling from $p(c)$ is efficient.
A sufficient condition for satisfying these two criteria is that the error channel maps computational basis states to other computational states, i.e., the channel's Kraus operators $K_m$ satisfy
\begin{equation}
K_m \ket{i} \propto \ket{i'},
\label{eq:computational_errors}
\end{equation} 
for all $m$, and where $\ket{i}, \ket{i'} \in \{\ket{0},\ket{1},\ket{W} \}$ are computational basis states.
The first criterion is satisfied because Eq.~(\ref{eq:computational_errors}) guarantees that a QRAM circuit interspersed with applications of the Kraus operators $K_m$ is still ECS. The second criterion is satisfied because the distribution $p(c)$ can be sampled efficiently by applying errors independently to each router (with appropriate probability) at each time step as the simulation proceeds.
In detail, suppose that at time $t$ the system is in a state $\ket{\psi(t)}$ that is a superposition of polynomially-many computational basis states, 
\begin{equation}
\ket{\psi(t)} = \sum_{\{i_1, \ldots i_{N-1}\}\in C} \alpha_i \ket{i_1, i_2, \ldots ,i_{N-1} },
\end{equation}
where $\ket{i_r}$ denotes the state of router $r$, and the cardinality of the set $C$ is $O(\mathrm{poly}N)$. The probability that a Kraus operator $K_m$ is applied to router $r$ is 
\begin{equation}
\mathrm{Tr} \left[K_m^\dagger K_m \rho_r \right],
\end{equation}
where $\rho_r(t) = \mathrm{Tr}_{\bar{r}} (\ket{\psi(t)}\bra{\psi(t)})$ is the reduced density matrix of router $r$, with $\mathrm{Tr}_{\bar{r}}$ denoting the partial trace over the rest of the system. \Cref{eq:computational_errors} guarantees that this probability is efficiently computable, so sampling from the possible errors at time $t$ is also efficient. This sampling procedure is repeated at each time step in order to sample from the full error configuration. 

\begin{figure*}
\centering{}\includegraphics[width=2.0\columnwidth]{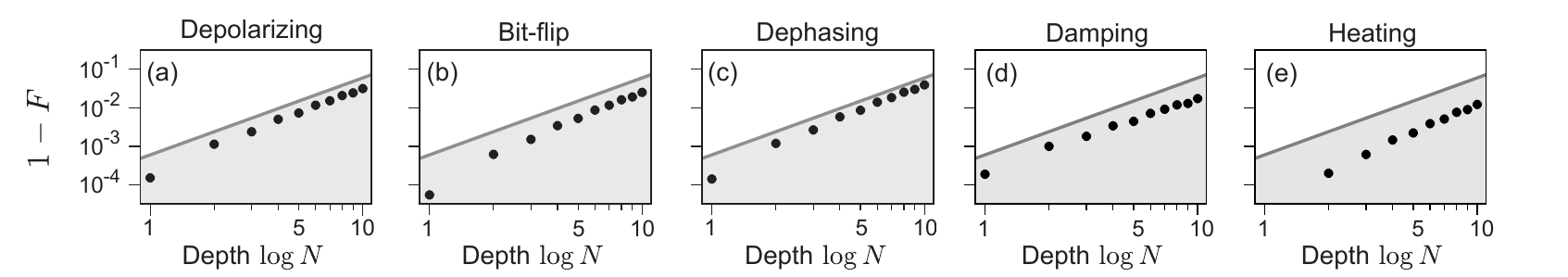}\caption{\label{fig3} 
Favorable error scaling. For a variety of error channels, the query infidelity (black dots) is calculated numerically and plotted as a function of the tree depth $\log N$ (note the logarithmic scaling on both axes).  The region defined by the upper bound (\ref{eq:infid_bound}) is shown in gray in each plot. 
Plotted infidelities are averages over many randomly generated binary data sets $\{x_0,\ldots x_{N-1}\}$. Each such data set is generated by randomly choosing each $x_i$ to be 0 or 1 with equal probability. Error bars are smaller than the dot size. The error rate for all plots is $\varepsilon = 10^{-4}$.
}
\end{figure*}

We apply this algorithm in order to compute the query infidelity for QRAM circuits with routers subject to a variety of noise channels. The results (\Cref{fig3}) confirm that the QRAM query infidelity scales favorably in the presence of realistic noise channels acting on all of the memory's components. We stress that, for such channels, the expected number of errors generally scales linearly with $N$. Results for qutrit depolarizing, bit-flip, and dephasing channels are shown in panels (a), (b), and (c), respectively (see \Cref{appendix:channels} for Kraus decompositions for each channel). These channels are all of the form (\ref{eq:mixed_unitary_channel}), so the query fidelity is subject to the bound (\ref{eq:infid_bound}). The numerical results are all clearly consistent with this bound, and the expected $1-F\propto \log^2 N$ scaling is evident on the log-log scale. In panels (d) and (e), we show numerical results for qutrit decay and heating channels (\Cref{appendix:channels}). We find that the query fidelities for these channels also satisfy the bound (\ref{eq:infid_bound}). Note, however, that the decay and heating channels are not mixed-unitary channels, so the query fidelities are subject to the general bound derived in \Cref{appendix:proof}, rather than \Cref{eq:infid_bound}.

%%%%%%%%%%%%%%%%%%%%%%%%%%%%%%%%%%%%%%%%%%%%%%%%%%
%%%%%%%%%%%%%%%%%%%%%%%%%%%%%%%%%%%%%%%%%%%%%%%%%%
%%%%%%%%%%%%%%%%%%%%%%%%%%%%%%%%%%%%%%%%%%%%%%%%%%

\section{Noise resilience with two-level routers}
\label{sec:5}

In \Cref{sec:3}, we proved that the query infidelity of the bucket-brigade QRAM scales favorably, even when inactive routers are subject to decoherence. It is thus natural to ask whether distinguishing between active and inactive routers is useful, and in fact whether the use of three-level routers is necessary in the first place. In this section, we show that the answer is no---the query infidelity still scales only {polylogarithmically} for QRAMs constructed from noisy two-level routers.
{As in \Cref{sec:3}, the argument presented to justify this claim is based on a careful analysis of how errors propagate.} Furthermore, we show that this same argument also reveals that noise resilience persists when the QRAM is initialized in an arbitrary state, and when the routing circuit [\Cref{fig1}(b)] is modified. Taken together, the results in this section show that the noise resilience of the bucket-brigade scheme is a robust property that is insensitive to implementation details. They also show that existing experimental proposals~\cite{hann2019,hong2012} employing two-level routers are noise-resilient.

\begin{figure}
\centering{}\includegraphics[width=\columnwidth]{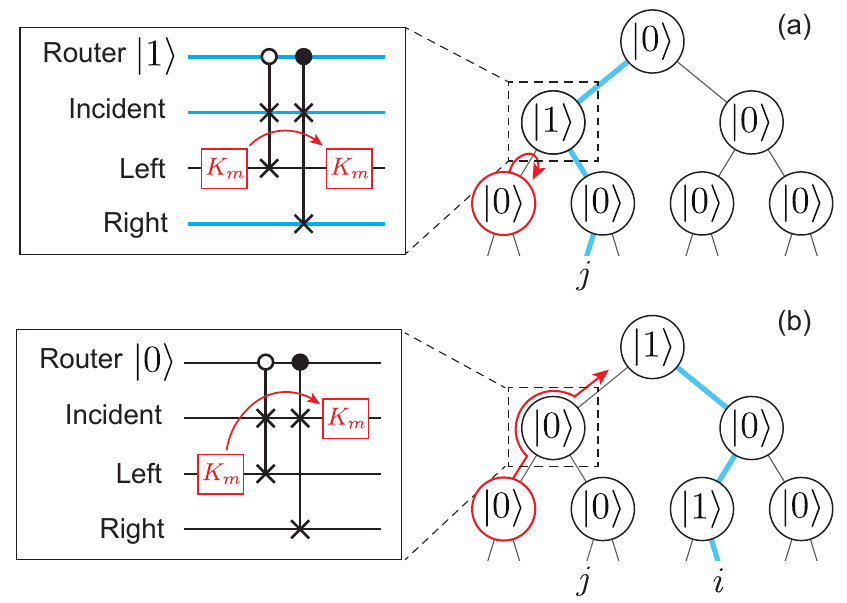}\caption{\label{fig4} Error propagation with two-level routers. (a) A query to memory element $j\in g(c)$, with an error $K_{m>0}$ applied to the red-outlined router. The circuit on the left shows how the error propagates through the router indicated by the dashed box. In this case, the error does not propagate into branch $j$. (b) A query to a different memory element $i\in g(c)$. In this case, the error propagates upward into branch $j$, in contrast to the situation in (a). 
}
\end{figure}

Consider a QRAM constructed from routers with only two states: $\ket{0}$ (route left) and $\ket{1}$ (route right). Routers are thus always active.   For concreteness, we suppose that the routing operation is implemented using the circuit in \Cref{fig1}(b), and that all routers are initialized in $\ket{0}$, though these assumptions can be relaxed. Unfortunately, the proof from \Cref{sec:3} cannot be directly applied to show that the query fidelity also scales favorably in this case. The proof fails in the case of two-level routers because the propagation of errors is no longer so highly constrained. 
Recall that in the case of three-level routers, errors do not propagate from bad branches into good branches. More precisely, for any $i,j \in g(c)$, errors do not propagate into branch $j$ when branch $i$ is queried. This is not the case for two-level routers: while errors do not propagate into branch $i$ when branch $i$ is queried, they can propagate into other branches $j$, as illustrated in \Cref{fig4}. Because of this difference, when multiple memory elements $i, j, \ldots \in g(c)$ are queried in superposition, it is not guaranteed that the address and bus registers will be disentangled from the routers at the end of the query.  Thus, \Cref{eq:final_system_state,eq:final_system_state_good} no longer hold. Instead, the final state $\ket{\Omega(c)}$ is given by 
\begin{equation}
\ket{\Omega(c)}=    
 \sum_{i\in g(c)} \alpha_i \ket{i}^A\ket{x_i}^B\ket{f_i(c)}^R + \ket{\mathrm{bad}(c)}, 
\end{equation}
where $\ket{f_i(c)}$ denotes the now address-dependent final state of the routers, and $\ket{f_i(c)} \neq \ket{f_j(c)}$ in general.  
As a result, the $i,j\in g(c)$ terms are no longer guaranteed to be in coherent superposition after tracing out the routers. Rather, the final state of the address-bus system is liable to contain an incoherent mixture of these terms. That is, the final density matrix can contain terms of the form $\ket{i,x_i}\bra{i,x_i}$ and $\ket{j,x_j}\bra{j,x_j}$ without $\ket{i,x_i}\bra{j,x_j}$ or $\ket{j,x_j}\bra{i,x_i}$ terms.  This loss of coherence reduces the fidelity.

We now proceed to estimate this reduction in fidelity. We find that the reduction is mild, such that the infidelity still scales only polylogarithmically with the memory size. Our approach is to isolate the subset of branches in $g(c)$ for which the sort of damaging error propagation described above does not occur. Explicitly, we define the subset $\tilde g(c) \subseteq g(c)$ as the largest subset such that for any $i,j\in \tilde g(c)$ errors do not propagate into branch $j$ during a query to element $i$. We then have that
$\ket{f_i(c)} = \ket{f_j(c)}$ by the same argument as given in \Cref{sec:3}. It follows that, if multiple memory elements in $\tilde g(c)$ are queried in superposition, the address and bus registers will be disentangled from the routers at the end of the query.

Having defined $\tilde g(c)$ as the subset of good branches without damaging error propagation, we are free to define all other branches as bad and then proceed exactly as in \Cref{sec:3}. In particular, we analogously define
\begin{equation}
\tilde \Lambda(c) =  \sum_{i\in \tilde g(c)} |\alpha_i|^2
\end{equation}
as the weighted fraction of good branches, and 
\begin{equation}
    F \geq [ 2 \mathbb E(\tilde \Lambda) - 1 ]^2,
\end{equation}
follows as the analog of \Cref{eq:infid_bound_intermediate}. Because $\tilde g(c) \subseteq g(c)$, we have that
\begin{equation}
    \mathbb E(\tilde \Lambda) = (1-\delta) \mathbb E
    (\Lambda),
\end{equation}
for some $\delta \in [0,1]$ to be determined. Proceeding as in \Cref{sec:3}, it follows that the infidelity satisfies the bound
\begin{equation}
    \label{eq:two_level_bound_intermediate}
    1-F\leq 4\varepsilon T \log N + 4\delta
\end{equation}
assuming $\varepsilon{T\log N} +\delta \leq 1/4$.

We can estimate $\delta$ by computing the average probability that errors propagate from bad branches into good branches. More specifically, we compute the probability that an error propagates into a branch $i \in g(c)$ when some other branch $j \in g(c)$ is queried.
Suppose that a router $r$ suffers an error at time step $t$, and let $P_{r\rightarrow i}(t)$ denote the probability of this error propagating into branch $i$. Then to leading order in $\varepsilon$, 
\begin{equation}
    \delta = \varepsilon \sum_{r,t}  P_{r\rightarrow i}(t) + O(\varepsilon ^2),
\end{equation}
which can be understood as the total probability that an error occurs and propagates into branch $i$.
To compute $\sum_{r,t}  P_{r\rightarrow i}(t)$ to leading order, we observe that
errors are generally free to propagate from a router's left output to its input, as illustrated in \Cref{fig4}(b). This is because, by default, all routers are initialized in $\ket{0}$, for which the routing operation swaps the states at the incident and left ports. In contrast, for an error to propagate upward from a router's right output, an additional error would be required to flip the router from $\ket{0}$ to $\ket{1}$. Thus, only the errors which can reach branch $i$ by propagating upward exclusively through the left outputs of routers contribute to $\sum_{r,t}  P_{r\rightarrow i}(t)$ to leading order in $\varepsilon$. A conservative overestimate is thus obtained by first enumerating all routers $r$ that are connected to $i$ through the left ports of other routers, then pessimistically taking $P_{r\rightarrow i}(t) = 1$ for each. There are at most $\log^2 N$ such routers, so 
\begin{equation}
    \delta \leq \varepsilon T \log^2 N +O(\epsilon^2).
\end{equation}
Substituting this expression into \Cref{eq:two_level_bound_intermediate}, we obtain
\begin{equation}
    \label{eq:two_level_bound}
 1-F \lesssim 4\varepsilon T \left(\log N +\log^2 N\right).
\end{equation}
Here we use the symbol $\lesssim$ to contrast this bound with \Cref{eq:infid_bound}; we proved the bound (\ref{eq:infid_bound}) rigorously, while we have obtained \Cref{eq:two_level_bound} through a scaling argument. As such, it is appropriate to focus only on the scaling of \Cref{eq:two_level_bound}. We see that the infidelity still scales only polylogarithmically with the memory size, indicating that a bucket-brigade QRAM constructed from noisy two-level routers also exhibits noise resilience. Note, however, that the infidelity here scales with $\log^3 N$ [recall $T= O(\log N$)], as opposed to $\log^2 N$ in the case of three-level routers. Both scalings are still favorable according to our definition, but the discrepancy indicates that three-level routers impart better noise resilience than two-level routers.

We simulate noisy QRAM circuits with two-level routers in order to verify this noise resilience. Simulation results are shown in \Cref{fig5}. For all noise channels simulated, the query infidelity is observed to scale polylogarithmically with the memory size, as expected. Moreover, the observed scaling exponents are $\leq3$ in all cases, consistent with the pessimistic $1-F\sim\log^3 N$ scaling given above.

\begin{figure*}[htbp]
\centering{}\includegraphics[width=2.0\columnwidth]{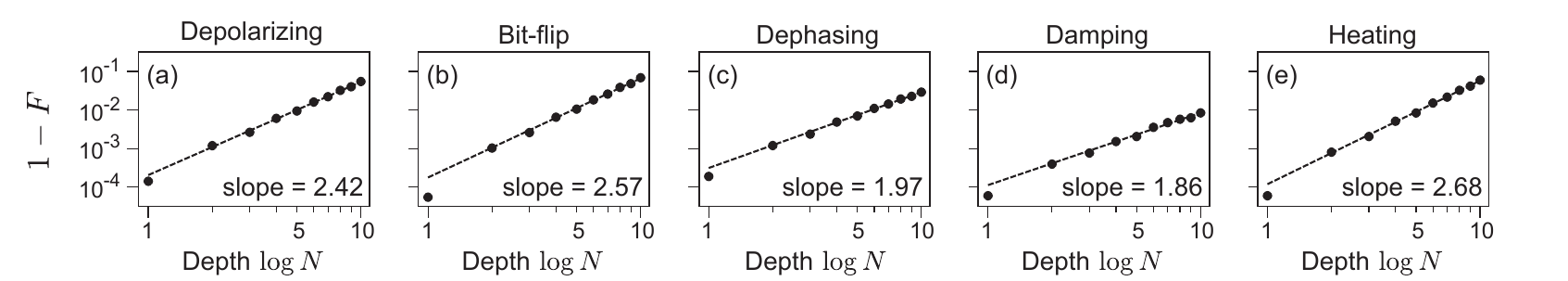}\caption{\label{fig5} Favorable error scaling with two-level routers. For a variety of error channels, the query infidelity (black dots) is calculated numerically and plotted as a function of the tree depth $\log N$.  Linear fits for each data set are shown as dashed lines, with the corresponding slopes given on each plot. Fits are performed only on data points with $\log N \geq 3$ so that the slopes are not skewed by finite-size effects at small $\log N$. Slopes $\leq3$ are consistent with the scaling argument in the text. The error rate for all plots is $\varepsilon = 10^{-4}$.
}
\end{figure*}

It is interesting to note that two-level routers are more resilient to certain noise channels than others, as quantified by the observed differences in scaling exponents. 
For example, the infidelity under the dephasing channel is observed to scale approximately as $1-F\sim \log^{2} N$.
{  
This relatively mild scaling can be explained as follows.
When the dephasing errors are propagated through the QRAM circuit, they may act non-trivially on the final state of the address and bus registers, but they act trivially on the final state of the routers (the all-$\ket{0}$ state).
As a result, the final state of the routers is the same for every address: $\ket{f_i(c)} = \ket{f_j(c)}$ for all $i,j$. Hence, $\tilde g(c) = g(c)$, and the bound from \Cref{sec:3} applies.} For the other channels, $\tilde g(c)\neq g(c)$ in general, consistent with observed scaling exponents $>2$. The case of amplitude damping is also interesting to consider: the expected number of errors for this channel is only $\varepsilon T \log N$ because only $\log N$ excitations are injected into the tree. Because $T = O(\log N)$, one expects the infidelity to scale with $\log^2 N$. The observed slope of $1.86$ is somewhat smaller owing to the fact that, in our simulations, excitations are only susceptible to damping while they reside in the tree.

The scaling argument presented in this section also suffices to show that the noise resilience persists in two other interesting situations: when the QRAM is initialized in an arbitrary state, and when the routing circuit is modified. 
Regarding initialization, observe that the above argument is straightforwardly modified to cover the case where all routers are initialized in $\ket{1}$ rather than $\ket{0}$. Indeed, such an argument holds regardless of whether a given router is initialized in $\ket{0}$ or $\ket{1}$.
It follows that the query infidelity scales favorably when the QRAM is initialized in an arbitrary state~\footnote{This observation is distinct from the observation of Refs.~\cite{low2018,berry2019} that the ancillary qubits used to perform a query can be ``dirty.'' See \Cref{appendix:copying_data}.} (though some additional care must be taken when copying data to the bus---see \Cref{appendix:copying_data} for details). 
This observation has great practical utility, as it means that QRAM can be constructed even from physical components that cannot reliably be initialized to a particular state. 

Regarding modifications to the routing circuit, it is helpful to consider an example. 
In Ref.~\cite{hann2019} a modified routing circuit was proposed in which one of the controlled-SWAP gates in \Cref{fig1}(b) is replaced by a SWAP gate. This modification has nontrivial effects on how errors propagate. With the modified circuit, errors can propagate from bad branches into good branches even when three-level routers are used. However, this is the same sort of damaging error propagation as is illustrated in \Cref{fig4}. Indeed, from the perspective of error propagation, the effect of this modification to the routing circuit is equivalent to replacing three-level routers with two-level routers. Accordingly, the argument above can be directly applied to show that the favorable scaling persists with the modified circuit.
This example demonstrates that noise resilience is not a specific feature of the routing circuit [\Cref{fig1}(b)]. 

Taken together, the results from this section demonstrate that the noise resilience of the bucket brigade architecture is a robust property that is insensitive to implementation details. 
This observation affords a great deal of freedom to experimentalists in deciding how the routers and routing operations could be implemented in practice.

%%%%%%%%%%%%%%%%%%%%%%%%%%%%%%%%%%%%%%%%%%%%%%%%%%
%%%%%%%%%%%%%%%%%%%%%%%%%%%%%%%%%%%%%%%%%%%%%%%%%%
%%%%%%%%%%%%%%%%%%%%%%%%%%%%%%%%%%%%%%%%%%%%%%%%%%

\section{Hybrid architectures}
\label{sec:hybrid}

The bucket-brigade architecture allows one to perform queries in $O(\log N)$ time using $O(N)$ qubits. This allocation of resources represents one extreme; at the other extreme are architectures~\cite{babbush2018,park2019,deveras2020} that perform queries in $O(N \log N)$ time using $O(\log N)$ qubits. 
In fact, there exists a family of architectures that interpolate between these two extremes to leverage this space-time trade-off~\cite{low2018,berry2019,dimatteo2020,paler2020a}. 
We refer to these as \emph{hybrid architectures}, and in this section we study their noise resilience. 
We find that hybrid architectures can be imbued with a partial noise resilience when they employ the bucket-brigade QRAM as a subroutine. 
As a result, these hybrid bucket-brigade architectures can have significantly higher query fidelities than other architectures that require the same resources. 

One of the primary benefits of the bucket-brigade architecture is that queries can be performed in only $O(\log N)$ time. These fast query times are essential for algorithms that must rapidly load large classical data sets in order to claim exponential speedups over their classical counterparts, e.g., quantum machine learning algorithms~\cite{lloyd2013,wittek2014,adcock2015,biamonte2017,ciliberto2018}. However, the $O(N)$ hardware overhead that enables such fast queries is practically daunting, and not all algorithms require such fast queries in the first place. In algorithms that only require comparatively small data sets to be loaded, e.g.~simulating local Hamiltonians~\cite{berry2012,berry2015,berry2015b,babbush2018,low2019,bauer2020b}, slower query times can be sufficient. Circuits that use fewer qubits at the price of longer query times are better suited for such algorithms.

\begin{figure}[tb]
\centering{}\includegraphics[width=1.0\columnwidth]{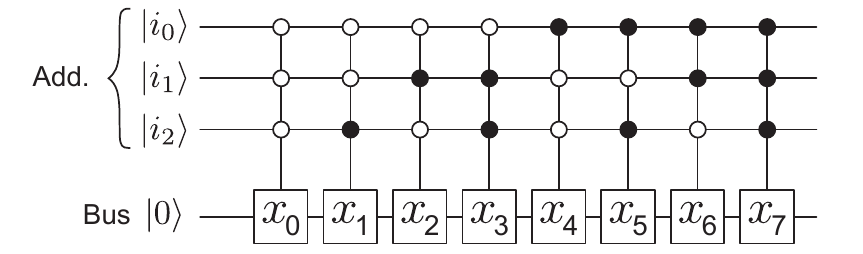}\caption{\label{fig:QROM} QROM circuit. The circuit implements operation (\ref{eq:QRAM_def}) by iterating over all $N$ possible states of the address register. The $j$-th gate flips the state of the bus qubit if the address register (Add.) is in state $\ket{j}$ and $x_j=1$, otherwise the gate acts trivially. 
}
\end{figure}

\Cref{fig:QROM} provides a straightforward example of such a circuit. To query a memory of size $N$, a sequence of $N$ multiply-controlled Toffoli gates is applied, where each gate has $\log N$ controls (the address qubits) and one target (the bus qubit). The circuit sequentially iterates over all $N$ possible addresses, flipping the bus qubit conditioned on the corresponding classical data. The circuit requires only $O(\log N)$ qubits, but it has depth $O(N \log N)$, since each multiply-controlled Toffoli gate can be performed in depth $O(\log N)$~\cite{saeedi2013}. 
Adopting the nomenclature introduced in Ref.~\cite{babbush2018}, we refer to such circuits as Quantum Read-Only Memory (QROM). We note that this circuit can be further optimized to reduce the depth, as shown in Ref.~\cite{babbush2018}. We omit these optimizations for simplicity, as they do not affect our main conclusions concerning the effects of noise.

\begin{figure*}[htbp]
\centering{}\includegraphics[width=1.75\columnwidth]{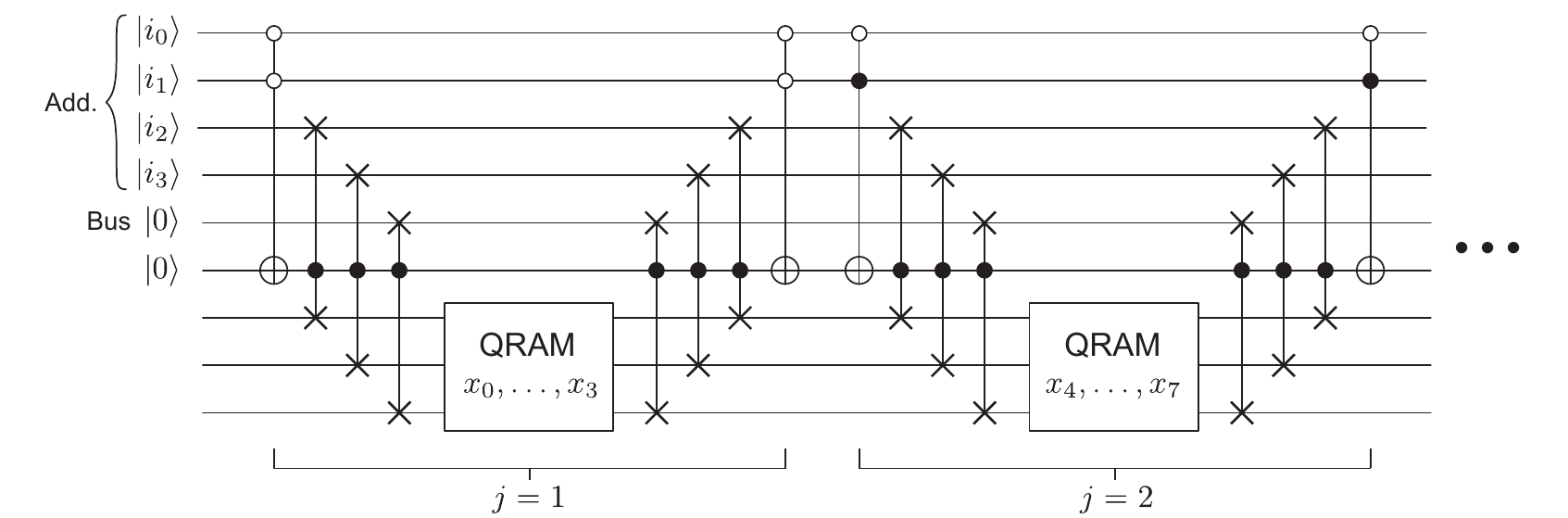}\caption{\label{fig:hybrid} Hybrid circuit. All $M=2^m$ possible states of the first $m$ address qubits are iterated over sequentially, as in QROM. Conditioned on these qubits, the remaining address qubits are used to query an $(N/M)$-cell classical memory via QRAM. In the circuit shown, $\log N =4$ and $m = 2$. The boxes labelled QRAM implement~(\ref{eq:QRAM_def}), using either the fanout or bucket-brigade architecture. At the $j$-th iteration ($j\in[1,M]$), the data elements $\{x_{[(j-1)N/M]},\ldots,x_{[j (N/M) -1]}\}$ are queried by the QRAM. Only the first two iterations are shown. The circuit depth is $O(M\log N)$, and the circuit uses $O(N/M+\log N)$ qubits, which includes the $O(N/M)$ ancillary qubits required by the QRAM (not shown).
The initial state of the QRAM can be arbitrary (see \Cref{sec:5}). }
\end{figure*}

More generally, circuits can be constructed that trade longer query times for fewer qubits by combining QROM and QRAM, as shown in~\Cref{fig:hybrid}. 
We introduce a tunable parameter $M \leq N$, defined to be a power of $2$. That is, $M = 2^m$, with $m$ an integer in the interval $[0, \log N]$. 
The idea is to divide the full classical memory into $M$ blocks, each with $N/M$ entries. 
These blocks are queried one-by-one using a QRAM of size $N/M$ concatenated with a QROM-like iteration scheme. 
The total hardware cost of the scheme is $O(\log N + N/M)$, comprising $O(\log N)$ qubits for the address and bus registers and $O(N/M)$ ancillary qubits for the QRAM. 
The total circuit depth is $O(M \log N)$ because each of the $M$ iterations in the circuit can be performed in depth $O(\log N)$. 
Therefore, by tuning the parameter $M$, one can interpolate between large-width, small-depth circuits like QRAM, and small-width, large-depth circuits like QROM. 
The hybrid circuit reduces to QRAM for $M = 1$, and QROM for $M = N$. 
We note that the circuits we introduce in \Cref{fig:hybrid} are very similar to those in Refs.~\cite{low2018,berry2019,dimatteo2020,paler2020a}. The main difference is that our circuits explicitly invoke QRAM as a subroutine, which makes the analysis of their noise resilience more straightforward.

Let us consider the effects of noise on these circuits. We first consider QROM, then turn to the hybrid circuits.
One can easily observe that QROM does not possess any intrinsic noise resilience. 
For example, when all memory elements are queried in equal superposition [$\alpha_i = 1/\sqrt{N}$ in \Cref{eq:QRAM_def}], a single dephasing error at any location in the QROM circuit reduces the query fidelity to 0. 
The effects of bit flips are similarly detrimental, assuming there is no contrived redundancy in the classical data. 
More generally, we can follow the approach of~\Cref{sec:3} and express the QROM query fidelity as $F = \sum_c p(c) F(c)$, where the error configuration $c$ specifies which Kraus operators are applied at each location in the circuit, and $F(c)$ is the final state fidelity of the address and bus registers given configuration $c$. 
In the case of QROM, only the error configuration with no errors is guaranteed to have unit or near-unit fidelity in general \footnote{Some other error configurations may have high fidelity for specific choices of the error channel, the initial address state, or the classical data, but we ignore this possibility to keep the analysis general and pessimistic.}. 
There are $O(N \log^2 N)$ possible error locations, so it follows that the QROM query infidelity scales as
\begin{equation}
1-F_{\text{QROM}} \sim \varepsilon N \log^2 N,
\end{equation}
to leading order.
Therefore, QROM is not noise resilient, since near-unit query fidelities generally require $\varepsilon  \ll 1/N$, neglecting logarithmic factors.

Similarly, the hybrid circuits do not exhibit noise resilience when the QRAM subroutines are implemented with the fanout architecture.
Recall from \Cref{sec:2} that the fanout architecture is not noise resilient; only the fanout's no-error configuration is guaranteed to have high fidelity in general. Because neither QROM nor the fanout QRAM are noise resilient, only the no-error configuration of the hybrid fanout circuit is guaranteed to have high fidelity. Since the number of possible error locations is $O(M \log N (\log N + N/M))$, the query fidelity scales as
\begin{align}
1-F_{\mathrm{hybrid,fanout}} &\sim \varepsilon (N \log N + M \log^2 N),
\end{align}
to leading order. 
Here again, error rates $\varepsilon \ll 1/N$ are required for near-unit query fidelity, neglecting logarithmic factors.

In contrast, the hybrid circuits do exhibit partial noise resilience when the QRAM subroutines are implemented with the bucket-brigade architecture. 
Because the bucket-brigade QRAM is resilient to noise, error configurations with errors occurring exclusively in the QRAM subroutines can still have high fidelities. 
We can obtain a lower bound on the query fidelity by neglecting all other configurations. Doing so allows us to bound the query fidelity by a product of two factors
\begin{equation}
F_{\mathrm{hybrid,BB}} \gtrsim (1-\varepsilon)^{O( M \log^2 N)} \times (1-\varepsilon)^{O(M \log N \log N/M)}
\end{equation}
The first factor is simply the probability that no errors occur outside the QRAM. 
The second factor is the expected fraction of error-free branches within the QRAM (each branch contains $\log N/M$ routers, and there are $T = O(M \log N)$ possible time steps at which errors may occur). We have related this expected fraction to $F_{\mathrm{hybrid-BB}}$ by the same argument as in \Cref{sec:3}. 
Thus, to leading order,
\begin{align}
1-F_{\mathrm{hybrid,BB}} &\lesssim \varepsilon M \log N (\log N + \log N/M), \\
&\sim \varepsilon M \log^2 N.
\end{align}
Note that we have not kept track of prefactors since we are only interested in how the infidelity scales; a strict upper bound could be rigorously derived following the approach of \Cref{sec:3}. 
Near-unit query fidelities only require error rates $\varepsilon  \ll 1/M$, neglecting logarithmic factors (cf.~the $\varepsilon \ll 1/N$ requirement for the other cases). 
Because $M\leq N$, the infidelity of the hybrid bucket-brigade architecture scales more favorably than both QROM and the hybrid fanout architecture. 
Of course, the extent of the scaling advantage depends on $M$. For example, if one chooses $M = \sqrt{N}$, so that the number of qubits and circuit depth are comparable, then the hybrid bucket-brigade architecture yields a quadratic improvement in the infidelity scaling. Note that we assume three-level routers above for simplicity; for two-level routers, one should replace $\log N/M \rightarrow \log^2 N/M$ in the above expressions, in accordance with the argument from~\Cref{sec:5}.

\section{Error-corrected QRAM}
\label{sec:6}

In this section, we show that the benefits of the bucket brigade scheme persist when quantum error correction is used. When the bucket-brigade QRAM is implemented using error-corrected routers and fault-tolerant routing operations~\cite{nielsen2000,gottesman2009}, the logical query infidelity scales only polylogarithmically with the memory size. Thus, error-corrected implementations of the bucket-brigade scheme can offer improved fidelity or reduced overhead relative to other implementations. In practice, these improvements may be tempered by the overhead associated with the fault-tolerant implementation of the routing operations, and we discuss the utility of the bucket-brigade architecture in light of such considerations. 

While we have shown that the query infidelity of the bucket-brigade scheme scales favorably with the memory size, strategies to further suppress the infidelity are desirable, and quantum error correction provides one possible approach. Indeed, error correction may be required in cases where the physical error rate cannot be made sufficiently small, or when many queries must be performed in sequence. For example, Ref.~\cite{arunachalam2015} argued that error correction is likely to be needed for any algorithm that requires a number of QRAM queries that scales superpolynomially in $\log N$, e.g., Grover's algorithm~\cite{grover1996}. 

It is thus natural to ask whether an error-corrected bucket brigade QRAM offers any advantages over other architectures. Indeed, this question was previously considered in Ref.~\cite{arunachalam2015}, where the authors argue in the negative. Their argument is based on the canonical attribution~\cite{giovannetti2008,giovannetti2008a,hong2012,arunachalam2015,ciliberto2018,adcock2015} of the bucket brigade's noise resilience to the limited number of active routers. Error-corrected routers must be considered active, they argue, and so the number of active routers is the same in both the fanout and bucket brigade schemes. Hence, the bucket brigade scheme was not believed to provide any advantage if error correction were used. 

As we have shown, however, the noise resilience of the bucket brigade scheme is not a function of the number of active routers, but rather a function of the limited entanglement among the routers. As a direct corollary of this result, we find that, in fact, the benefits of the bucket brigade scheme do persist when error correction is used. The proof from~\Cref{sec:3} is agnostic to whether the routers are composed of uncorrected physical qubits or error-corrected logical qubits, provided that uncorrectable logical errors occur independently with some probability $\varepsilon_L$ (which can be guaranteed by implementing the routing operations fault tolerantly). 
Physical errors occurring with probability $\varepsilon$ can simply be replaced by logical errors occurring with probability $\varepsilon_L$, and one obtains the corresponding bound
\begin{equation}
\label{eq:logical_scaling}
1-F_L \leq  4\varepsilon_L T_L \log N,
\end{equation}
where $F_L$ is the query fidelity of the logical QRAM circuit, and $T_L$ is the circuit depth. Thus, when implemented fault-tolerantly, the logical bucket-brigade circuits possess an intrinsic resilience to logical errors, in that the logical infidelity scales only polylogarithmically with the the size of the memory (This scaling assumes $T_L=O(\log N)$; see further discussion at the end of this section). 

To provide further exposition, we give a concrete example of an error-corrected quantum router. Consider a quantum error correcting-code, with logical codewords $\ket{0_L}$ and $\ket{1_L}$ satisfying the Knill-Laflamme conditions~\cite{knill1997,nielsen2000}, 
\begin{equation}
\label{eq:knill_laflamme}
P K_i^\dagger K_j P = h_{ij} P,
\end{equation}
where $P$ is the projector onto the code space, the $\{ K_i\}$ are the set of correctable errors, and $h$ is a Hermitian matrix. 
A logical two-level quantum router then constitutes a single logical qubit (similarly, a logical three-level quantum router can be constructed from a pair of logical qubits, for example). 
%can then be constructed from two logical qubits. For example, we define
% \begin{align}
%     \ket{\overline{W}} &\equiv \ket{0_L}\otimes\ket{0_L}, \\
%     \ket{\overline{0}} &\equiv \ket{1_L}\otimes\ket{0_L}, \\
%     \ket{\overline{1}} &\equiv \ket{1_L}\otimes\ket{1_L},
% \end{align}
% and we refer to the collection of states $\{\ket{\overline{W}},\ket{\overline{0}},\ket{\overline{1}}\}$ as a logical router. 
%Errors in the physical qubits comprising the logical router are corrected by independently correcting the two logical qubits.
Crucially, the logical routers comprising the QRAM can be corrected without revealing any information about which memory elements are being accessed. This is because the conditions~(\ref{eq:knill_laflamme}) guarantee that errors can be corrected without revealing any information about the encoded state. Even when the logical router is in a superposition of different states, or entangled with other routers, syndrome measurements do not reveal information about the router state.
Note that the conditions~(\ref{eq:knill_laflamme}) also guarantee that information is not leaked to the environment; the states $\ket{0_L}$ and $\ket{1_L}$ necessarily have equal probability of suffering errors. 

Because of the favorable logical error scaling~\Cref{eq:logical_scaling}, error-corrected implementations of the bucket-brigade scheme can offer improved fidelity or reduced overhead relative to other implementations. For instance, if the same error-correcting code is used in fault-tolerant implementations of the bucket-brigade and fanout QRAMs, the logical infidelity of bucket brigade QRAM will be lower than the logical infidelity of the fanout QRAM by a factor of $\sim 1/N$ in general.
Alternatively, if a given application requires that QRAM have a logical infidelity below some threshold, the error-correction overhead required to realize such high-fidelity queries can be significantly smaller for the bucket-brigade scheme relative to the fanout scheme. Indeed, even if the reduction in error-correction overhead is fairly small for each router, the total overhead reduction considering all $N$ routers can be significant. Such reductions could be of significant practical benefit.
{
For context, we note that detailed overhead estimates for fault-tolerant QRAM using the surface code were made in Ref.~\cite{dimatteo2020}; these overheads can potentially be improved by exploiting the bucket-brigade's noise resilience.
}

We conclude this section with an important caveat concerning fault-tolerant QRAM.

\section{Discussion}
\label{sec:7}

We have shown that the bucket-brigade QRAM architecture possesses a remarkable resilience to noise. Even when all $O(N)$ components comprising the QRAM are subject to arbitrary error channels, the query infidelity scales only polylogarithmically with the memory size.  
As a result, the bucket-brigade architecture can be used to perform high-fidelity queries of large memories without the need for quantum error correction, provided physical error rates are low.
Importantly, we prove that this noise resilience holds for \emph{arbitrary} error channels, demonstrating that a noise-resilient QRAM can be implemented with realistically noisy devices. 

In the near-term, this noise resilience could facilitate experimental demonstrations and benchmarking of numerous quantum algorithms. 
We are presently in the Noisy, Intermediate-Scale Quantum (NISQ) era~\cite{preskill2018}, when making more qubits is easier than making better qubits. The same is likely to be true even in the era of early fault-tolerance. 
In these eras, the bucket-brigade architecture---with its larger overhead and noise resilience---could actually prove to be more practical than alternatives like QROM (see \Cref{sec:hybrid}) that have a lower overhead but are less tolerant to noise.
The bucket-brigade architecture thus more readily enables small-scale, near-term implementations of algorithms, and important practical insights are likely to be gained from such demonstrations.  
{Schemes to further suppress the query fidelity without resorting to full error correction~\cite{leeinpreparation} could prove useful in this effort.}

\begin{table*}[t]
    \centering
    \def\arraystretch{1.5}
    \begin{tabular}{ |c|c|c|c|c|  }
\hline
$N$ & $Q$ & Applicable architectures & QEC required?  & Paradigmatic example  \\
\hhline{|=|=|=|=|=|}
$\exp(n)$ & $\exp(n)$ &  QRAM & Yes  & Searching an unstructured database~\cite{grover1996} \\
\hline
$\exp(n)$ & $\text{poly}(n)$ & QRAM & Maybe not  & 
 \makecell{Solving linear systems of equations~\cite{harrow2009} \\ (sparse, well-conditioned systems)}\\%Solving linear systems of equations~\cite{harrow2009} \\
%\hline
%$\text{poly}(n)$ & $\exp(n)$ & QRAM, QROM, Hybrid & Yes  & ?  \\
\hline
$\text{poly}(n)$ & $\text{poly}(n)$ & QRAM, QROM, Hybrid & Maybe not  & Simulating local Hamiltonians~\cite{childs2018}  \\
\hline
\end{tabular}
    \caption{Algorithm categorization. Algorithms are sorted based on how the size of the classical memory, $N$, and the number of queries, $Q$, scale with the number of qubits, $n$. When $N = \exp(n) $, QRAM is the only suitable architecture, assuming $\text{poly}(n)$ query times are required. When $Q = \text{poly}(n)$ quantum error correction may not be required, depending on the physical error rates. For the examples in the last two rows, $Q$ also depends on the particular algorithm used and the desired precision; we assume these are chosen such that $Q = \text{poly(n)}$. We omit the case of $N=\text{poly}(n)$ and $Q = \exp(n)$, for which the query complexity is exponential in the problem size.
  }
    \label{tab:algorithm_categories}
\end{table*}

In the long-term, this noise resilience may prove useful in facilitating speedups for certain quantum algorithms, but it is important that the required resources be carefully assessed before a speedup via QRAM is claimed.  
Consider an oracle-based algorithm that requires $n$ qubits (not including ancillary qubits needed to implement the oracle).
As we show in \Cref{tab:algorithm_categories}, such algorithms can be conveniently classified according to how the size of the classical memory being queried, $N$, and the total number of queries, $Q$, scale with $n$. 
Assuming $\text{poly}(n)$ query times are required, the memory size $N$ dictates whether QRAM (as opposed to QROM or a hybrid architecture) is required to implement the oracle. The number of queries $Q$ dictates whether error correction is necessarily required~\cite{arunachalam2015}.
The noise resilience of the bucket-brigade has the biggest potential impact in case of $N= \exp(n)$ and $Q = \text{poly}(n)$. In this case, QRAM is required, and the noise-resilience of the bucket-brigade architecture, together with the comparatively small number of queries, allows for the possibility that the QRAM could be implemented without error correction. 
Of course, the noise resilience can also be advantageous in the other cases, where hybrid architectures may be employed (\Cref{sec:hybrid}) or when error correction is used (\Cref{sec:6}).

Finally, it is worth emphasizing that the results in this paper constitute general statements about the bucket-brigade architecture, independent of its application to particular algorithms. In fact, the architecture may prove useful in applications other than facilitating algorithmic speedups. For example, Ref.~\cite{giovannetti2008b} employs the bucket-brigade architecture in a quantum cryptographic protocol. The architecture may similarly prove useful for quantum communication or metrology. Exploring applications of the bucket-brigade architecture---and the utility of its noise resilience---in these other contexts represents an interesting direction for future research. In particular, applications involving quantum queries of quantum data remain largely unexplored. 

\section{Acknowledgements}
We thank Shruti Puri and Xiaodi Wu helpful discussions. We acknowledge the Yale Center for Research Computing for use of its high-performance computing clusters. CTH~acknowledges support from the NSF Graduate Research Fellowship Program (DGE1752134). We acknowledge support from the ARO (W911NF-18-1-0020, W911NF-18-1-0212), ARO MURI (W911NF-16-1-0349), AFOSR MURI (FA9550-19-1-0399), DOE (DE-SC0019406), NSF (EFMA-1640959, OMA-1936118, EEC-1941583), NTT Research, and the Packard Foundation (2013-39273).

\normalbaselines
\bibliography{QRAM_bibliography_final2}

\newpage
\appendix

\section{Bucket-brigade QRAM circuit}
\label{appendix:circuit}

\begin{figure*}
\centering{}\includegraphics[width=2.0\columnwidth]{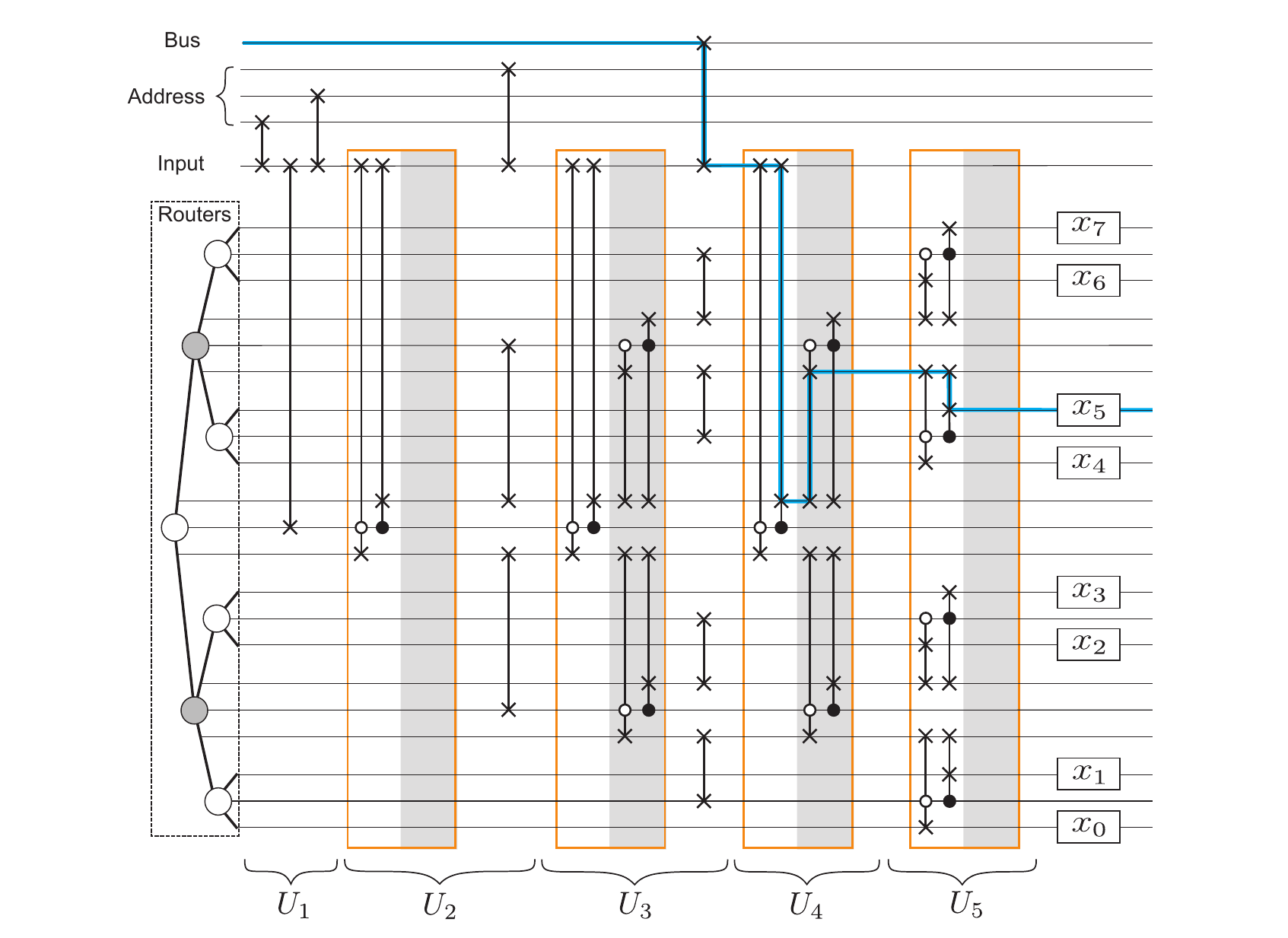}\caption{\label{figS1} Bucket-brigade QRAM circuit for $N=8$. The bus and address register are indicated by rails at the top of the diagram, and the routers are indicated by the rails below. The ``input'' rail is an extra ancilla, included only to simplify the circuit diagram.
For each router shown on the left, there are three rails: one for the router's internal state, and two for the router's two output modes. 
All rails represent qubits (qutrits) when the variant of the bucket-brigade architecture with two-level (three-level) routers is used; the circuit is the same in either case.
To simplify the circuit we use a shorthand notation for the routing operation (top right).
Each orange box contains routing operations for a subset of routers in the tree, with routing operations for the even levels first (white background), followed by the odd levels (gray background). 
Though only a subset of the possible routing operations are actually applied at each time step, in principle all routing operations at even (odd) levels can be applied in parallel. 
The first part of the circuit, $U_3 U_2 U_1$, routes the addresses into the tree, as in \Cref{fig1}(d). The next part, $U_5 U_4$, routes the bus to the appropriate memory cell. The path of the bus is highlighted in blue for the case where the three address qubits are initialized to $\ket{010}$. The last layer of gates copy data into the state of the bus (see \Cref{appendix:copying_data}). To route the addresses and bus out of the tree and complete the query, the inverse operation, $(U_5 U_4 U_3 U_2 U_1)^\dagger$, must subsequently be applied (not shown).
}
\end{figure*}

\Cref{figS1} shows a circuit diagram for the bucket-brigade QRAM in the case of an $N=8$ cell memory. In contrast to the $T=O(\log^2 N)$ depth circuits described implicitly in Refs.~\cite{giovannetti2008,giovannetti2008a,arunachalam2015}, the circuit depth here is only $T=O(\log N)$.  The quadratic speedup is due to additional parallelization of the routing operations that we now describe. 

In Refs.~\cite{giovannetti2008,giovannetti2008a,arunachalam2015}, address qubits are routed into the tree one at a time; the $(\ell+1)$-th address qubit is only injected into the tree once the $\ell$-th address qubit has taken its position at level $\ell$ of the tree. If each routing operation takes one time step, then one waits $\ell$ time steps between the injection of the $\ell$-th and $(\ell +1)$-th address qubits.  The total circuit depth is obtained by summing the number of time steps that it takes for each address to reach the corresponding level,
\begin{equation}
    T \sim \sum_{\ell = 0}^{\log N - 1} \ell = O(\log^2 N).
\end{equation}
However, it is not necessary to wait for an address qubit to reach its destination before subsequent address qubits are sent into the tree, and this realization enables the circuit depth to be reduced to $O(\log N)$. Notice that the routing operations for routers located at even levels of the tree act on mutually disjoint qubits and hence mutually commute (the same is true for the odd levels). Thus, all routing operations at either even or odd levels can be performed in parallel. In practice, then, once an address qubit has reached level $\ell = 2$, the next address qubit can be sent into the tree at level $\ell = 0$, and the two can be routed in parallel. This way, the wait time between the injection of subsequent address qubits into the tree is constant (cf.~the $O(\ell)$ wait time above). Exploiting this additional parallelism, the total circuit depth is reduced to $T=O(\log N)$. 

The circuit diagram in \Cref{figS1} illustrates how such $T=O(\log N)$ bucket-brigade QRAM circuits are structured. The circuit is a sequence of $T$ constant-depth circuits $U_t$. During each block $U_t$, a subset of the possible routing operations is performed (orange boxes in the figure), with routing operations at even levels performed before routing operations at odd levels. Importantly, because the operations at even levels can be performed in parallel (sim.~for odd levels), each orange box constitutes only two layers of parallel gates.
In principle, new address qubit can then be injected into the tree at the beginning of each orange box, with routing performed in parallel whenever multiple address qubits are being routed through the tree. The case of $N=8$ shown in the figure is too small to demonstrate this parallelism, but to see how the parallelism would manifest at larger memory sizes, notice that it would be possible to inject another qubit into the tree during $U_5$ and route this qubit (at level $\ell = 0$) in parallel with one at level $\ell = 2$.

\section{Copying data to the bus}

\label{appendix:copying_data}

In this section, we explicitly describe how classical data can be copied into the state of the bus. Slightly different procedures are required depending on whether the QRAM is implemented with two-level or three-level systems, and whether the QRAM is initialized in a known state or in some arbitrary state (see \Cref{sec:5}).

We begin with the case where the QRAM is implemented with two-level routers, as described in \Cref{sec:5}. 
Each router's incident and output modes are also taken to be physical two-level systems. 
All routers and their respective modes are initialized to $\ket{0}$.
For reasons that will become apparent shortly, we suppose that the bus qubit is initialized to $\ket{+} \equiv (\ket{0} + \ket{1})/\sqrt{2}$ prior to the query. During the query, this bus qubit is routed down the tree, to an output mode of some router at the bottom level.  At this point, classical data is encoded into the state of the bus qubit by applying classically controlled $Z$ gates, as illustrated in \Cref{fig:copy}(a). If the memory element being queried is 1, a $Z$ gate is applied, and the state of the bus is flipped from $\ket{+}$ to $\ket{-}\equiv(\ket{0}-\ket{1})/\sqrt{2}$. If the memory element queried is 0, no $Z$ gate is applied, and the bus remains in $\ket{+}$. In this way, the classical bit is encoded in the $\ket{\pm}$ basis of the bus qubit. 
Note, however, that because the location of the bus is not known, classically controlled $Z$ gates must be applied to the output modes of all routers at the bottom level of the tree.

\begin{figure}
\centering{}\includegraphics[width=1.0\columnwidth]{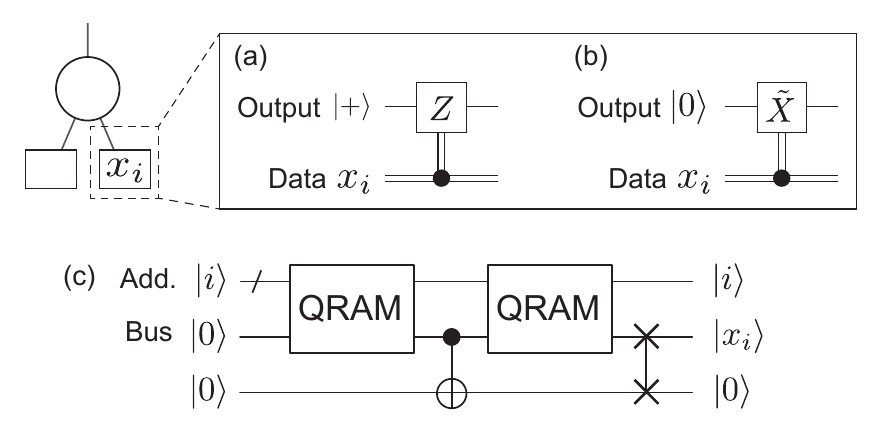}\caption{\label{fig:copy} Circuits for copying classical data. (a) Two-level circuit. The bus qubit is encoded within a physical two-level system and initialized in $\ket{+}$. A $Z$ gate flips the bus to $\ket{-}$ conditioned on the classical data. (b) Three-level circuit. The bus qubit is encoded within a two-level subspace of a physical three-level system and initialized in $\ket{0}$. The $\tilde{X}$ gate (see text) flips the bus to $\ket{1}$ conditioned on the classical data. (c) Query circuit for QRAM initialized in an arbitrary state. The circuits assumes three-level routers, so the bus is initialized in $\ket{0}$ and circuit (b) is employed within each QRAM block to copy data to the bus. An analogous circuit can be constructed for two-level routers. The ancillary qubits comprising the QRAM's routers (not shown) can be initialized in an arbitrary state.
}
\end{figure}

This data copying operation has a crucial property, which we call \emph{no extra copying}: in the absence of errors, the copying operation acts trivially on all modes that do not contain the bus qubit.
In the above case, all modes that do not contain the bus are in $\ket{0}$, so they are unaffected by the $Z$ gates, hence why we use the $\ket{\pm}$ basis for the bus~\cite{hann2019}.
The no extra copying property is crucial because it guarantees that the final state of the tree is the same across all good (error-free) branches, as required by the arguments in the main text. 
Were this property not to hold, the final state of the tree would depend on which element was queried, so the bus would remain entangled with the routers after the query, even in the absence of errors. 

Now let us consider the case where the QRAM is implemented with three-level routers, as described in \Cref{sec:2,sec:3}.
Each router's incident and output modes are taken to be physical three-level systems, whose basis states we also label as $\ket{0}$, $\ket{1}$, and $\ket{W}$.
The address and bus qubits are encoded within the $\ket{0,1}$ subspace of such three-level systems.
Prior to the query, all routers, as well as their incident and output modes, are initialized to $\ket{W}$, and the bus is initialized to $\ket{0}$. During the query, the bus is routed to an output mode of some router at the bottom level of the tree. Data is copied into the bus by applying classically controlled $\tilde X$ gates to the output modes [\Cref{fig:copy}(b)], where
\begin{equation}
    \tilde X = \ket{1}\bra{0} + \ket{0}\bra{1} + \ket{W}\bra{W}.
\end{equation}
If the memory element being queried is 1, the $\tilde X$ gate is applied, and the state of the bus is flipped from $\ket{0}$ to $\ket{1}$. If the memory element queried is 0, no $\tilde X$ gate is applied, and the bus remains in $\ket{0}$. In this way, the classical bit is encoded in the $\ket{0,1}$ basis of the bus qubit (one could also choose to encode the information in the $\ket{\pm}$ basis by constructing an analogous $\tilde Z$ gate).
Here again, the classically controlled gates must be applied to the output modes of all routers at the bottom of the tree. This operation satisfies the no extra copying property because, in the absence of errors, all modes not containing the bus are in $\ket{W}$, on which $\tilde X$ acts trivially.

In order to enforce the no extra copying property, both of the above data copying operations rely on the fact that the routers, as well as their input and output modes, are initialized to some known state. When the QRAM is initialized in an arbitrary state (see \Cref{sec:5}), however, additional care must be taken to ensure this property still holds. 
The challenge is that the mode that actually contains the bus must somehow be distinguished from all the other modes, which may have been initialized in the same state as the bus. 
This problem is solved by the circuit in \Cref{fig:copy}(c).
The QRAM is queried twice, and the no extra copying property is guaranteed by the fact that the entire QRAM unitary operation is idempotent. In particular, even if the process of copying data during the first query acts non-trivially on modes not containing the bus, these modes are always reset to their initial states by the second query. 
In fact, even the bus is reset to its initial state by the second query. Thus, the information stored in the bus is copied to an ancillary qubit in between the two queries, then swapped back into the bus after the second query. We emphasize that the query fidelity of this circuit scales favorably, which can be shown by simply replacing $T\rightarrow 2T$ in the scaling argument from \Cref{sec:5} to account for the fact that the QRAM is called twice.

As an aside, let us distinguish between our observation that QRAM is resilient to noise even when initialized in an arbitrary state (\Cref{sec:5}), and the observation of Refs.~\cite{low2018,berry2019} that the ancillary qubits used to perform a query can be ``dirty.'' The latter states that circuits can be designed such that, \emph{in the absence of errors}, any ancillary qubits used during the query are returned to their initial state after the query, regardless of what the initial state was (note the circuit in \Cref{fig:copy}(c) has this property).  
In contrast, our observation concerns what happens when errors occur during the query: the query infidelity of the circuit \Cref{fig:copy}(c) scales favorably even when the QRAM is initialized in an arbitrary state.

\section{Error channel Kraus decompositions}
\label{appendix:channels}
In this Appendix, we give the Kraus decompositions for the channels used in our simulations. We specify a generic channel $\mathcal{E}$ as via a list of its Kraus operators as
\begin{equation}
    \mathcal E  = \left\{K_0, K_1, K_2, \ldots \right\}.
\end{equation}

\subsection{Qubit error channels}
Let $X,Y,Z$ denote the Pauli matrices. The decompositions of the qubit error channels are
\begin{align}
\text{Depolarizing} &= \left\{
\sqrt{1-\varepsilon} I, 
\sqrt{\frac{\varepsilon}{3}} X, \sqrt{\frac{\varepsilon}{3}} Y,
\sqrt{\frac{\varepsilon}{3}} Z
\right\} \\
\text{Bit-flip} &= \left\{
\sqrt{1-\varepsilon} I, 
\sqrt{\varepsilon} X
\right\} \\
\text{Dephasing} &= \left\{
\sqrt{1-\varepsilon} I, 
\sqrt{\varepsilon} Z
\right\} \\
\text{Damping} &= \left\{ 
\ket{0}\bra{0} +\sqrt{1-\varepsilon}\ket{1}\bra{1}, 
\sqrt{\varepsilon} \ket{0}\bra{1}
\right\} \\
\text{Heating} &= \left\{
\ket{1}\bra{1} +\sqrt{1-\varepsilon}\ket{0}\bra{0}, 
\sqrt{\varepsilon} \ket{1}\bra{0}
\right\} 
\end{align}

\subsection{Qutrit error channels}
We follow the definitions for qutrit depolarizing, bit-flip, and dephasing channels given in Refs.~\cite{ramzan2012}, \cite{arunachalam2015}, and \cite{wei2013}, respectively. Define the operators
\begin{equation}
    A_1 = \begin{pmatrix}
        0 & 1 & 0\\
        0 & 0 & 1\\
        1 & 0 & 0
        \end{pmatrix},\, 
     A_2 = \begin{pmatrix}
        1 & 0 & 0\\
        0 & \omega & 0\\
        0 & 0 & \omega^2
        \end{pmatrix},
    % A_3 = \begin{pmatrix}
    %     1 & 0 & 0\\
    %     0 & \omega^2 & 0\\
    %     0 & 0 & \omega
    %     \end{pmatrix} 
\end{equation}
where the matrices are written in the $\{\ket{W},\ket{0},\ket{1} \}$ basis, and $\omega = e^{i 2\pi/3}$. 
The decompositions of the qutrit error channels are
\begin{widetext}
\begin{align}
\text{Depolarizing} &= \bigg\{
\sqrt{1-\varepsilon} I, 
\sqrt{\frac{\varepsilon}{8}} A_1, \sqrt{\frac{\varepsilon}{8}} A_2,
\sqrt{\frac{\varepsilon}{8}} A_1^2, 
\sqrt{\frac{\varepsilon}{8}} A_2^2,
\sqrt{\frac{\varepsilon}{8}} A_1A_2, 
\sqrt{\frac{\varepsilon}{8}} A_1^2A_2, 
\sqrt{\frac{\varepsilon}{8}} A_1A_2^2,
\sqrt{\frac{\varepsilon}{8}} A_1^2A_2^2
\bigg\} \\
\text{Bit-flip} &= \left\{
\sqrt{1-\varepsilon} I, 
\sqrt{\varepsilon} \left(\ket{0}\bra{1} + \ket{1}\bra{0}\right)
\right\} \\
\text{Dephasing} &= \left\{
\sqrt{1-\varepsilon} I, 
\sqrt{\frac{\varepsilon}{2}} A_2,
\sqrt{\frac{\varepsilon}{2}} A_2^2
\right\} \\
\text{Damping} &= \bigg\{ 
\ket{W}\bra{W} +\sqrt{1-\varepsilon}\left(\ket{0}\bra{0}+\ket{1}\bra{1}\right), 
\sqrt{\varepsilon} \ket{W}\bra{0},
\sqrt{\varepsilon} \ket{W}\bra{1}
\bigg\} \\
\text{Heating} &= \bigg\{ 
\ket{0}\bra{0}+\ket{1}\bra{1}
 +\sqrt{1-\varepsilon}\ket{W}\bra{W},  
\sqrt{\frac{\varepsilon}{2}} \ket{0}\bra{W},
\sqrt{\frac{\varepsilon}{2}} \ket{1}\bra{W}
\bigg\}.
\end{align}
\end{widetext}

\section{Proof of noise resilience for arbitrary error channels}
\label{appendix:proof}
In this Appendix, we prove that, for arbitrary error channels,
the QRAM query infidelity satisfies the bound 
\begin{equation}
    1-F \leq A  \varepsilon T \log N,
\end{equation}
where $A$ is a constant of order 1 (defined below). 
We begin by defining the error model and introducing some convenient notation.

%%%%%%%%%%%%%%%%%%%%%%%%%%%%%%%%%%%%%%%%%%%%%%%%
\subsection{Error model}
\label{sec:proof_error_model}
%%%%%%%%%%%%%%%%%%%%%%%%%%%%%%%%%%%%%%%%%%%%%%%%

We suppose that at each time step, every router in the QRAM is subject to an error channel $\mathcal E$ of the form 
\begin{equation}
    \rho \rightarrow \mathcal E (\rho) = \sum_m K_m \rho K_m^\dagger.
\end{equation}
One could also consider situations where different routers are subject to different error channels; the proof straightforwardly extends to such situations.
For the moment, we make two additional assumptions that simplify the proof. First, we restrict our attention to channels with Kraus rank two, i.e.~channels that can be expressed with only two non-zero Kraus operators, $K_0$ and $K_1$. Second, we assume that each router is subjected to the error channel $\mathcal E$ only once, at the time step $t^*$ after all addresses have been routed into place and immediately before the bus enters the tree (this time step is the one illustrated in \Cref{fig1,fig1pt5}). We relax both of these assumptions later on.

Let us define the error rate, $\varepsilon$. For a generic state $\ket{\psi}$, we may write
\begin{equation}
    K_0\ket{\psi} = a_\psi \ket{\psi} + b_\psi \ket{\psi^\perp}
\end{equation}
where $\braket{\psi|\psi^\perp} = 0$, and $a_\psi$,$b_\psi$ are complex numbers. We define $\varepsilon$ as the smallest positive real number such that, for any $\ket{\psi}$, we have
\begin{align}
    \label{eq:b_psi}
    |b_{\psi}|&\leq \sqrt{\varepsilon},
\end{align}
and for any collection of arbitrary states $\ket{\psi_i}$, we also have
\begin{align}
    \label{eq:a_psi}
    \Re\left[\prod_{i=1}^{m} a_{\psi_i}\right] &\geq (1-\varepsilon)^{m/2}, \text{ for any $m\leq \log N$}.
\end{align}
This definition is unconventional, so let us emphasize that $\varepsilon$ can be understood as quantifying the distance between $\mathcal E$ and the identity channel. Specifically, for any given $\varepsilon_0 \in [0,1]$, there always exists some $\delta\in[0,1]$ such that, whenever $|\mathcal E - I|<\delta$ according to some distance metric, then $\varepsilon < \varepsilon_0$.

We adopt this definition of $\varepsilon$ because it turns out to be very convenient for the proof. For example, it follows that
\begin{align}
    \braket{\psi|K_0^\dagger K_0|\psi} &\geq 1-\varepsilon \\
    \braket{\psi|K_1^\dagger K_1|\psi} &\leq \varepsilon.
\end{align}
and, for a generic product state, 
$\ket{\xi_m} = \bigotimes_{i = 1}^{m} \ket{\psi_i} $,
\begin{align}
\label{eq:ideal_components}
    K_0^{\otimes m} \ket{\xi_m}  = a_{\xi_m} \ket{\xi_m} + b_{\xi_m} \ket{\xi_m^\perp},
\end{align}
where it follows from \Cref{eq:b_psi,eq:a_psi} that
\begin{align}
    \label{eq:ideal_components_a}
    \Re(a_{\xi_m} )&\geq (1-\varepsilon)^{m/2} \\
    \label{eq:ideal_components_b}
    |b_{\xi_m}| &\leq \varepsilon^{m/2}
\end{align}
for $m\leq \log N$. 

%%%%%%%%%%%%%%%%%%%%%%%%%%%%%%%%%%%%%%%%%%%%%%%%
\subsection{Notation}
%%%%%%%%%%%%%%%%%%%%%%%%%%%%%%%%%%%%%%%%%%%%%%%%
We now define some convenient notation.
As in the main text, we use the shorthand
\begin{equation}
    K_c \equiv \bigotimes_{r=1}^{N-1} K_{c(r)} 
\end{equation}
to denote the composite Kraus operator acting on all routers in the tree. Here, $K_{c(r)}\in\{K_0,K_1\}$ denotes the Kraus operator applied to router $r$ at time $t^*$ as specified by the error configuration $c$. 
For the $i$-th branch of the tree, we also define the operator $ K_{c \not\in i}$ as that obtained by taking $K_{c}$ and replacing $K_{c(r)}\rightarrow I$ for any router $r$ in branch $i$:
\begin{equation}
    K_{c \not\in i} \equiv \bigotimes_{r=1}^{N-1} 
    \begin{cases}
        I, & r\in i,  \\
        K_{c(r)}, & \text{otherwise.} \\ 
    \end{cases}
\end{equation}
Similarly, for a pair of branches, $i$ and $j$, we analogously define
\begin{equation}
    \label{eq:Kc_not_in_ij}
    K_{c \not\in i,j} \equiv \bigotimes_{r=1}^{N-1} 
    \begin{cases}
        I, & r\in i \cup j,  \\
        K_{c(r)}, & \text{otherwise.} \\ 
    \end{cases}
\end{equation}
where $r\in i\cup j$ denotes the set of routers in either branch $i$ or $j$.

Additionally, we define $\ket{R_i}$ to be the ideal state of the routers at time $t^*$ assuming the $i$-th memory element is queried. Specifically, $\ket{R_i}$ is a computational basis state for which all routers $r\not\in i$ are in the state $\ket{W}$, and routers $r\in i$ are in either $\ket{0}$ or $\ket{1}$ and carve out a path to the memory element $i$. As examples, $\ket{R_{101}}$ is illustrated in \Cref{fig1}(d), and the superposition
\begin{equation*}
    \ket{R_{000}}+\ket{R_{010}}+\ket{R_{011}}+\ket{R_{101}}+\ket{R_{110}}  
\end{equation*}
is illustrated in \Cref{fig1pt5}. 

Finally, throughout the proof it will be convenient to work with unnormalized quantum states. We distinguish between normalized states $\ket{\psi}$ and unnormalized states $\ket{\overline\psi}$ with an overbar.

\begin{widetext}

%%%%%%%%%%%%%%%%%%%%%%%%%%%%%%%%%%%%%%%%%%%%%%%%
\subsection{Proof}
%%%%%%%%%%%%%%%%%%%%%%%%%%%%%%%%%%%%%%%%%%%%%%%%
Our proof of QRAM's noise resilience for generic error channels follows the same outline as the proof for mixed-unitary error channels given in the main text.
We again begin by defining the final state of the system $\Omega$ as an incoherent mixture over final states for different error configurations $c$,
\begin{equation}
    \Omega = \sum_c \overline\Omega(c)
\end{equation}
in analogy to \Cref{eq:final_mixture}. Here, the final system state given error configuration $c$ is $\overline\Omega(c) = \ket{\overline\Omega(c)}\bra{\overline\Omega(c)}$, where $\ket{\overline\Omega(c)}$ is the unnormalized state
\begin{equation}
     \ket{\overline \Omega(c)}= V\left[
\ket{0}^A \ket{0}^B K_c \left(\sum_i\alpha_i\ket{R_i}^R\right) 
\right],
\end{equation}
where the superscripts $A,B,R$ denote the addresses, bus, and routers, respectively. 
In accordance with the error model described in \Cref{sec:proof_error_model}, 
the quantity in square brackets is the ideal state of the system at time $t^*$ with errors $K_c$ subsequently applied to the routers. The unitary $V$ comprises all operations performed after time $t^*$. Namely, $V$ routes the bus to memory, copies data into the bus, then routes the bus and all address qubits out of the tree. Note that $K_c$ is not unitary, and the states $\ket{\overline \Omega(c)}$ are not normalized. The probability $p(c)$ of error configuration $c$ occurring is 
\begin{equation}
    \label{eq:pc}
    p(c) = \mathrm{Tr}[\overline \Omega(c)].
\end{equation}

The query fidelity is given by 
\begin{equation}
    F = \sum_c \overline F(c)
\end{equation}
where
\begin{equation}
\overline{F}(c) = \braket{\psi_\text{out}|
\text{Tr}_R\overline \Omega(c)
|\psi_\text{out}} 
\end{equation}
in analogy to \Cref{eq:fid_sum_over_c,eq:configuration_fidelity}. As in the main text, we proceed by placing a lower bound on $\overline F(c)$. To do so, we write
\begin{equation}
\label{eq:final_system_state_unnormalized}
\ket{\overline{\Omega}(c)}=
\ket{\overline{\text{good}} (c)} + \ket{\overline{\text{bad}} (c)},
    % \ket{\overline{\Omega}(c)} = V\left[ \ket{0}^A\otimes\ket{0}^B\otimes\left(\ket{\overline{\text{good}} (c)}^R + \ket{\overline{\text{bad}} (c)}^R\right)  \right] 
\end{equation}
in analogy to \Cref{eq:final_system_state}, and where the definitions of $\ket{\overline{\text{good}} (c)}$ and $\ket{\overline{\text{bad}} (c)}$ will be given shortly. 

Let us recall the definition of \emph{good} and \emph{bad} branches given in the main text. There, a branch $i$ was defined to be good if and only if $\mathbf{i}\cap \mathbf{c} = \varnothing$, where $\mathbf{i}$ is the set of routers in the $i$-th branch of the tree, and $\mathbf{c}$ is the set of routers which have $K_{m>0}$ applied to them according to error configuration $c$. We retain this definition here. Now, in main text we considered the case of mixed-unitary error channels, for which $K_0\propto I$, and for these channels we have that
\begin{equation}
    K_c\ket{R_i} = K_{c\not\in i} \ket{R_i},\, \text{for $K_0 \propto I$}
    \label{eq:K0_iden}
\end{equation}
for all $i\in g(c)$, where $g(c)$ denotes the set of good branches. In words, \Cref{eq:K0_iden} says that effect of the of no-error operators on the active routers is trivial. 
However, in the present case of general channels, for which $K_0\neq I$ in general, the no-error backaction associated with $K_0$ can alter the states of active routers. In particular,
\begin{equation}
    K_c\ket{R_i} = K_{c\not\in i}\left( a_{R_i} \ket{R_i} + b_{R_i}\ket{R_i^\perp}\right),\, \text{for general $K_0$}
    \label{eq:K0_not_iden}
\end{equation}
for all $i\in g(c)$, and where $\ket{R_i^\perp}$ denotes an orthogonal state, $\braket{R_i|R_i^\perp} = 0$. It follows from \Cref{eq:ideal_components,eq:ideal_components_a,eq:ideal_components_b} that the coefficients $a_{R_i}$ and $b_{R_i}$ satisfy
\begin{align}
    \label{eq:R_ideal_components_a}
    \Re\left(a_{R_i}\right) &\geq (1-\varepsilon)^{\log N/2} \\
    \label{eq:R_ideal_components_b}
    \left|b_{R_i}\right| &\leq \varepsilon^{\log N/2}.
\end{align}
In words, \Cref{eq:K0_not_iden} says that effect of the of no-error operators on the active routers is nontrivial and places the active routers in a superposition of the ideal state and some orthogonal state. \Cref{eq:R_ideal_components_a,eq:R_ideal_components_b} bound the coefficients in this superposition; in the relevant regime of $\varepsilon \log N \ll 1$, this nontrivial, no-error backaction is small, and the state $K_c\ket{R_i}$ is close to $K_{c\not\in i}\ket{R_i}$.

Now, let us define $\ket{\overline{\text{good}} (c)}$ and $\ket{\overline{\text{bad}} (c)}$ in \Cref{eq:final_system_state_unnormalized}. 
Given \Cref{eq:K0_not_iden}, it will be convenient to separate the ideal and non-ideal components of $K_c\ket{R_i}$, retaining only the ideal components in $\ket{\overline{\text{good}} (c)}$. We thus define
\begin{equation}
\ket{\overline{\text{good}}(c)} \equiv
    V\left[\ket{0}^A\ket{0}^B \sum_{i\in g(c)} \alpha_i a_{R_i} K_{c\not\in i}\ket{R_i}^R\right]
    % \ket{\overline{\text{good}}(c)} \equiv \sum_{i\in g(c)} \alpha_i\left( a_{R_i} K_{c\not\in i}\ket{R_i}\right)
\end{equation}
in analogy to \Cref{eq:final_system_state_good}, and
\begin{equation}
    \ket{\ubad} \equiv \ket{\overline \Omega(c)}-\ket{\ugood} =  V\left[
    \ket{0}^A  \ket{0}^B  \left(
    \sum_{i\in g(c)} \alpha_i b_{R_i} K_{c\not\in i}\ket{R_i^\perp}^R + \sum_{j\not\in g(c)} \alpha_j K_c \ket{R_i}^R\right)\right].
    % \ket{\overline{\text{bad}}(c)} \equiv \sum_{i\in g(c)} \alpha_i\left( b_{R_i} K_{c\not\in i}\ket{R_i^\perp}\right) + \sum_{j\not\in g(c)} \alpha_j \left(K_c \ket{R_i}\right)
\end{equation}
Note that $\ket{\overline{\text{bad}}(c)}$ not only contains a contribution from the bad branches, $j\not \in g(c)$, but also from the no-error backaction in the good branches, $i \in g(c)$.

We now seek to prove analogous statements to \Cref{eq:good_overlap,eq:bad_overlap} (given by \Cref{lem:2,lem:3} below). To do so, we first show that the overlap of $\ket{\overline{\mathrm{good}}(c)}$ and the ideal state is large in the following sense.
\begin{lemma}
\label{lem:good}
The overlap between $\ket{\overline{\text{good}} (c)}$ and the ideal final state, $\ket{\psi_\mathrm{out},f(c)}$, satisfies the bound
\begin{equation}
    % \Re\left[\braket{\psi_\text{out}^{AB},f(c)^{R}|V|0^A,0^B,\overline{\text{good}}(c)^R}\right] \geq \frac{(1-\varepsilon)^{k\log N}}{(1-\varepsilon_W)^{\log N}} \sqrt{q(c)}\, \Lambda(c)
   \Re\left[\braket{\psi_\mathrm{out},f(c)|\overline{\mathrm{good}}(c)}\right] \geq \sqrt{q(c)}\, \Lambda(c) \left(\frac{1-\varepsilon}{1-\varepsilon_W}\right)^{\frac{3}{2}\log N}
\end{equation}
where $\Lambda(c)= \sum_{i\in g(c)}|\alpha_i|^2$ is the weighted fraction of good branches, $\ket{f(c)}$ is some fixed final state of the routers (defined below), and 
\begin{equation}
    \label{eq:qc}
    q(c) \equiv (1-\varepsilon_W)^{(N-1)-|c|}\varepsilon_W^{|c|},
\end{equation}
with $\varepsilon_W \equiv 1 - \braket{W|K_0^\dagger K_0|W}$, and $|c|$ denotes the number of Kraus operators $K_{m>0}$ in the composite Kraus operator $K_c$, i.e.~$|c|$ is the number of errors.
\end{lemma}

\begin{proof}
Consider a good branch $i\in g(c)$. We have that
\begin{equation}
    \label{eq:v0A0B}
     V\left[\ket{0}^A\ket{0}^B K_{c\not\in i}\ket{R_i}^R\right] = \ket{i}^A  \ket{x_i}^B \ket{\overline{f}_i(c)}^R, 
\end{equation}
because all active routers in $K_{c\not\in i}\ket{R_i}^R$ are in their ideal states, so the query succeeds. Here,
\begin{equation}
    \ket{\overline{f}_i(c)} = \bra{i}^A\bra{x_i}^B V\left[\ket{0}^A\ket{0}^B K_{c\not\in i}\ket{R_i}^R\right]
\end{equation}
is the unnormalized final state of the routers associated with branch $i$. Consider a second good branch, $j\in g(c)$, with $j\neq i$. In general, the final state of the routers with respect to branch $j$ may be different from the final state with respect to branch $i$, i.e.~$\ket{\overline{f}_i(c)}\neq \ket{\overline{f}_j(c)}$. 
As a result, the routers may be entangled with the address and bus in $\ket{\ugood}$. 

However, we now show that the overlap $\braket{\overline{f}_i(c)|\overline{f}_j(c)}$ is large, such that the routers are \emph{nearly disentangled} from the address and bus in $\ket{\ugood}$. Recall the definition of $K_{c\not\in i,j}$ [\Cref{eq:Kc_not_in_ij}] as the operator obtained by taking $K_c$ and replacing all Kraus operators acting on routers in branches $i$ and $j$ with the identity. Observe that
\begin{equation}
    K_{c\not\in i} \ket{R_i} = K_{c\not\in i,j}\left(a'_{R_i} \ket{R_i} + b'_{R_i}\ket{R_i^\perp} \right)
\end{equation}
where it follows from \Cref{eq:ideal_components,eq:ideal_components_a,eq:ideal_components_b} that
\begin{align}
    \label{eq:R_ideal_components_a_prime}
    \Re\left(a'_{R_i}\right) &\geq (1-\varepsilon)^{\log N/2} \\
    \label{eq:R_ideal_components_b_prime}
    \left|b'_{R_i}\right| &\leq \varepsilon^{\log N/2}.
\end{align}
Now, the overlap is given by
\begin{align}
    \braket{\overline{f}_i(c)|\overline{f}_i(j)} &=
    \left[ 
    \left(\bra{0}^A\bra{0}^B\bra{R_i}^R K_{c\not\in i}^\dagger\right)V^\dagger \ket{i}^A\ket{x_i}^B
    \right]
    \left[ 
   \bra{j}^A\bra{x_j}^B V  \left(\ket{0}^A\ket{0}^B K_{c\not\in j}\ket{R_j} \right)
    \right] \\
    \label{eq:f_overlap1}
    & = (a'_{R_i})^{*}a'_{R_j}
    \left[ 
    \left(\bra{0}^A\bra{0}^B\bra{R_i}^R K_{c\not\in i,j}^\dagger\right)V^\dagger \ket{i}^A\ket{x_i}^B
    \right]
    \left[ 
   \bra{j}^A\bra{x_j}^B V  \left(\ket{0}^A\ket{0}^B K_{c\not\in i,j}\ket{R_j} \right)
    \right], \\
    \label{eq:f_overlap2}
    & = (a'_{R_i})^{*}a'_{R_j} \left|K_{c\not\in i,j} \ket{R_i} \right|^2.
\end{align}
To obtain \Cref{eq:f_overlap1} we have used the fact that 
\begin{equation}
    \bra{i}^A\bra{x_i}^B V\left(\ket{0}^A\ket{0}^B K_{c\not\in i,j}\ket{R_i^\perp}^R\right) = 0
\end{equation}
because by definition $K_{c\not\in i,j}\ket{R_i^\perp}$ has at least one active router in an orthogonal state relative to $\ket{R_i}$, so that after the application of $V$ the final address state is necessarily perpendicular to $\ket{i}$. To obtain \Cref{eq:f_overlap2} we have used the fact that the two states in brackets in \Cref{eq:f_overlap1} are in fact the same. This is because both $K_{c\not\in i,j}\ket{R_i}$ and $K_{c\not\in i,j}\ket{R_j}$ have routers in branches $i$ and $j$ in their ideal states. As articulated in the main text and illustrated in \Cref{fig3}, it follows that the propagation of errors is constrained in such a way that the final states of the routers (the states in brackets) are the same. 

Continuing with our calculation of $\braket{\overline{f}_i(c)|\overline{f}_j(c)}$, note that the norm $\left|K_{c\not\in i,j} \ket{R_i} \right|^2$ can be computed straightforwardly using the fact both that $K_{c\not\in i,j}$ and $\ket{R_i}$ are tensor products over the different routers. We obtain,
\begin{equation}
    \left|K_{c\not\in i,j} \ket{R_i} \right|^2 = \varepsilon_W^{|c|}(1-\varepsilon_W)^{(N-1)-2\log N-|c|},
    \label{eq:KcijRi}
\end{equation}
where $|c|$ is the number of Kraus operators $K_{m>0}$ in the composite Kraus operator $K_c$, i.e.~$|c|$ is the number of errors. We have used the definition $\braket{W|K_1^\dagger K_1|W} = \varepsilon_W$, which implies $\braket{W|K_0^\dagger K_0|W} = 1-\varepsilon_W$. Thus, combining \Cref{eq:R_ideal_components_a_prime,eq:f_overlap2,eq:KcijRi}  the real part of the overlap can be bounded as
\begin{equation}
    \Re\left[\braket{\overline{f}_i(c)|\overline{f}_j(c)}\right] \geq (1-\varepsilon)^{\log N} \varepsilon_W^{|c|}(1-\varepsilon_W)^{(N-1)-2\log N-|c|} = \frac{(1-\varepsilon)^{\log N}}{(1-\varepsilon_W)^{2\log N}} q(c),
\end{equation}
where $q(c)$ is defined in \Cref{eq:qc}.
More convenient for our purposes is the equivalent bound
\begin{equation}
    \label{eq:Refic}
    \Re\left[\braket{f_i(c)|\overline{f}_j(c)}\right] 
    \geq \frac{(1-\varepsilon)^{\log N}}{(1-\varepsilon_W)^{\frac{3}{2}\log N}} \sqrt{q(c)},
\end{equation}
where $\ket{f_i(c)}$ is the normalized version of $\ket{\overline{f}_i(c)}$, and we have used the fact that $\braket{\overline{f}_i(c)|\overline{f}_i(c)} = q(c)/(1-\varepsilon_W)^{\log N}$.

Now, for any $i\in g(c)$, consider the overlap, 
\begin{align}
    \braket{\psi_{\mathrm{out}},f_i(c)|\ugood} &= \left[ \sum_{j} \alpha_j^* \bra{j}^A\bra{x_j}^B \bra{f_i(c)}^R\right]V\left[ \sum_{k\in g(c)} \alpha_k a_{R_k}\ket{0}^A\ket{0}^B  K_{c\not\in k}\ket{R_k}^R\right] \\
    & = \left[ \sum_{j} \alpha_j^* \bra{j}^A\bra{x_j}^B \bra{f_i(c)}^R\right] \left[\sum_{k\in g(c)} \alpha_k a_{R_k} \ket{k}^A\ket{x_k}^B \ket{\overline{f}_k(c)}^R\right] \\ 
    & = \sum_{j\in g(c)} |\alpha_j|^2 a_{R_j} \braket{f_i(c)|\overline{f}_j(c)}
\end{align}
where the second line follows from \Cref{eq:v0A0B}. Using \Cref{eq:R_ideal_components_a,eq:Refic}, we can bound the real part of this overlap as
\begin{equation}
    \Re\left[\braket{\psi_\text{out},f_i(c)|\ugood}\right] \geq \sqrt{q(c)}\, \Lambda(c) \left(\frac{1-\varepsilon}{1-\varepsilon_W}\right)^{\frac{3}{2}\log N} ,
\end{equation}
completing the proof.
\end{proof}

The result of \Cref{lem:good} is not precisely analogous to \Cref{eq:good_overlap}. Rather, the analogous statement is \Cref{lem:2} (see below), which also incorporates detrimental effects of $\ket{\ubad}$. To incorporate these effects, it is convenient to write 
\begin{equation}
    \ket{\ubad} = \ket{\ubad^\parallel}+\ket{\ubad^\perp}
\end{equation}
where
\begin{align}
    \ket{\ubad^\perp} &\equiv \ket{\ubad} - \frac{\braket{\ugood|\ubad}}{\braket{\ugood|\ugood}} \ket{\ugood} \\
    \ket{\ubad^\parallel} &\equiv \ket{\ubad} - \ket{\ubad^\perp}
\end{align}
are the components of $\ket{\ubad}$ perpendicular and parallel to $\ket{\ugood}$, respectively. We proceed by first quantifying the detrimental effects of $\ket{\ubad^\parallel}$. 

\begin{lemma}
\label{lem:2}
The overlap of $\ket{\overline{\mathrm{good}}(c)} + \ket{\overline{\mathrm{bad}}(c)^{\parallel}}$ with the ideal state can be bounded as
\begin{equation}
\Re\left[\bra{\psi_\mathrm{out},f(c)}
    \left(
   \ket{\overline{\mathrm{good}}(c) }+ \ket{\overline{\mathrm{bad}}(c)^{\parallel}}
    \right)\right] \geq (1-B\varepsilon^2)  \sqrt{q(c)}\, \Lambda(c) \left(\frac{1-\varepsilon}{1-\varepsilon_W}\right)^{\frac{3}{2}\log N},
\end{equation}
with $B \approx 1$ in the relevant limit of $\varepsilon \log N \ll 1$. 
\end{lemma}
\begin{proof}
We prove the result by showing that the overlap of $\ket{\ubad}$ and $\ket{\ugood}$ is small relative to magnitude of $\ket{\ugood}$. Specifically, we prove that
\begin{equation}
    \label{eq:lemma2_alt}
    \frac{\left|\braket{\overline{\mathrm{good}}(c)|\overline{\mathrm{bad}}(c)}\right|}{\braket{\overline{\mathrm{good}}(c)|\overline{\mathrm{good}}(c)}} \leq B \varepsilon^2.
\end{equation}
Then, because $\ket{\ugood}+\ket{\ubad^\parallel}$ is parallel to $\ket{\ugood}$, and because its magnitude is bounded from below by $ (1-B\varepsilon^2)||\ket{\ugood}||$, the lemma follows.

First, we consider the magnitude of $\ket{\ugood}$,
\begin{align}
    \braket{\ugood|\ugood} &= \sum_{i,j\in g(c)} (\alpha_i a_{R_i})^* \alpha_j a_{R_j} \braket{R_i|K_{c\not\in i}^\dagger K_{c\not\in j}|R_j}\\
    &= \sum_{i\in g(c)} |\alpha_i a_{R_i}|^2 \braket{R_i|K_{c\not\in i}^\dagger K_{c\not\in i}|R_i}.
\end{align}
Using \Cref{eq:R_ideal_components_a} together with the fact that $\braket{R_i|K_{c\not\in i}^\dagger K_{c\not\in i}|R_i} = q(c)/(1-\varepsilon_W)^{\log N}$,
we can bound this overlap as
\begin{equation}
    \label{eq:56}
     \braket{\ugood|\ugood} \geq q(c) \Lambda(c) \left(\frac{1-\varepsilon}{1-\varepsilon_W}\right)^{\log N}.
\end{equation}
Next, we consider the overlap of $\ket{\ugood}$ and $\ket{\ubad}$, given by
\begin{align}
    \label{eq:ijgc0}
    &\braket{\ugood|\ubad} = 
    \left[ \sum_{i\in g(c)} (\alpha_i a_{R_i})^* \bra{R_i} K_{c\not\in i}^\dagger \right]
    \left[\sum_{j\in g(c)} \alpha_j b_{R_j} K_{c\not\in i}\ket{R_j^\perp} + \sum_{k\not\in g(c)} \alpha_k K_c \ket{R_k} \right] \\
    \label{eq:ijgc}
    & = \sum_{i\in g(c)}\sum_{j\in g(c), j\neq i} (\alpha_i a_{R_i})^* \alpha_j b_{R_j} \braket{R_i| K_{c\not\in i}^\dagger K_{c\not\in j}|R_j^\perp} 
    + \sum_{i\in g(c)}\sum_{k\not\in g(c)} (\alpha_i a_{R_i})^* \alpha_k  \braket{R_i| K_{c\not\in i}^\dagger K_{c}|R_k},
\end{align}
where we enforce $j\neq i$ in the second line by noting that $\braket{R_i|K_{c\not\in i}^\dagger K_{c\not\in i}|R_i^\perp} = 0$. 
\begin{remark}
\label{remark:1}
A subtle, but important point is that we can actually further enforce
\begin{align}
    j &\neq i + 1,\text{ for $i$ odd,}\\
    j &\neq i - 1,\text{ for $i$ even,}
\end{align}
in summation indexed by $j$ in \Cref{eq:ijgc}.
This is because, for $i$ even, branches $i$ and $i+1$ share the same active routers, and can be grouped together when using \Cref{eq:K0_not_iden} to define of $\ket{\ugood}$ and $\ket{\ubad}$. With these slightly modified definitions, $\ket{\ugood}$ contains the state $K_{c\not\in i}(\alpha_i\ket{R_i}+\alpha_{i+1}\ket{R_{i+1}})$, and $\ket{\ubad}$ contains the \emph{orthogonal} state $K_{c\not\in i}(\alpha_i\ket{R_i}+\alpha_{i+1}\ket{R_{i+1}})^{\perp}$. 
It follows that the $j = i+1$ terms vanish in \Cref{eq:ijgc} for even $i$. The $j = i-1$ terms vanish for odd $i$ by the same argument. To keep the notation simple, however, we will continue to simply write $j\neq i$ and reference this Remark whenever this  distinction is relevant.
\end{remark}
The first term in \Cref{eq:ijgc} can be simplified using the relation
\begin{equation}
   b_{R_j} K_{c\not\in j}\ket{R_j^\perp} = K_c\ket{R_j} - a_{R_j}K_{c\not\in j}\ket{R_j},
\end{equation}
which is equivalent to \Cref{eq:K0_not_iden}. Inserting this expression into \Cref{eq:ijgc} yields
\begin{align}
    \label{eq:good_bad_parallel_overlap}
    \braket{\ugood|\ubad^\parallel} = \sum_{i\in g(c)} \sum_{j\neq i} (\alpha_i a_{R_i})^* \alpha_j \braket{R_i| K_{c\not\in i}^\dagger K_c|R_j},
\end{align}
which we have simplified using the fact that $\braket{R_i|K_{c\not\in i}^\dagger K_{c\not\in j}|R_j} = \delta_{ij}$. \Cref{eq:good_bad_parallel_overlap} may be equivalently expressed as an inner product between two $N$-dimensional complex vectors, $\boldsymbol{\alpha_g}$ and $M\boldsymbol{\alpha}$,
\begin{equation}
    \braket{\ugood|\ubad^\parallel} = \braket{\boldsymbol{\alpha_g},M\boldsymbol{\alpha}}
\end{equation}
where the $i$-th entry of the vector $\boldsymbol{\alpha}$ is $\alpha_i$, the $i$-th entry of $\boldsymbol{\alpha_g}$ is
\begin{equation}
\alpha_{g,i} = \begin{cases} 
    a_{R_i} \alpha_i &\mbox{for } i\in g(c) \\
    0 &\mbox{for } i\not\in g(c),
    \end{cases}   
\end{equation}
and $M$ is an $N\times N$ complex matrix with entries $M_{ij} = \braket{R_i|K_{c\not\in i}^\dagger K_c|R_j}(1-\delta_{ij})$.

Consider the quantity $\braket{R_i| K_{c\not\in i}^\dagger K_c|R_j}$, with $i\in g(c)$ and $j\neq i$. Since all states and operators involved are tensor products, this term is a product of $N-1$ matrix elements, one for each of the $N-1$ routers. We now enumerate these matrix elements, first assuming $j\in g(c)$. Of the $N-1$ matrix elements, $(N-1) - 2\log N -|c|$ elements are $\braket{W|K_0^\dagger K_0|W} = 1-\varepsilon_W$, corresponding to error-free routers outside both branches $i$ and $j$. (There may be additional such elements, depending on how many routers $i$ and $j$ share, see \Cref{eq:61} below). With $j\in g(c)$, an additional $|c|$ matrix elements are $\braket{W|K_1^\dagger K_1|W} = \varepsilon_W$, corresponding to routers that suffer errors. The remaining $2\log N$ matrix elements can be enumerated as follows. For a given state $\ket{R_i}$, let $\ket{R^{(\ell)}_i}$ denote the corresponding state of the router at level $\ell$ (1-indexed) in branch $i$ of the tree. Additionally, let $\ell_{ij}$ denote the smallest value $\ell$ for which $\braket{R^{(\ell)}_i|R^{(\ell)}_j} = 0$.
At each level $\ell$ of the tree, there are two remaining matrix elements---associated with routers at level $\ell$---that we have not yet enumerated. These matrix elements are
\begin{align}
    \label{eq:61}
    \braket{R_i^{(\ell)}| K_0|R_j^{(\ell)}}\braket{W|K_0^\dagger K_0|W}, &\text{ for $\ell \leq \ell_{ij}$,} \\ 
    \label{eq:62}
    \braket{R_i^{(\ell)}|K_0|W}\braket{W|K_0^\dagger K_0|R_j^{(\ell)}}, &\text{ for $\ell > \ell_{ij}$.}
\end{align}
The absolute values of these products can be bounded using \Cref{eq:b_psi,eq:a_psi},
\begin{align}
    \label{eq:63}
    \left|\braket{R_i^{(\ell)}| K_0|R_j^{(\ell)}}\braket{W|K_0^\dagger K_0|W}\right| &\leq
    \begin{cases} 
    (1-\varepsilon_W) &\mbox{for } \ell \neq \ell_{ij} \\
    (1-\varepsilon_W) \, \sqrt{\varepsilon} &\mbox{for } \ell = \ell_{ij}
    \end{cases}
    \\ 
    \label{eq:64}
    \left|\braket{R_i^{(\ell)}| K_0|W}\braket{W|K_0^\dagger K_0|R_j^{(\ell)}}\right| &\leq \varepsilon_W \sqrt{\varepsilon},
\end{align}
where the bound on the first line is tighter for $\ell =\ell_{ij}$ because $\braket{R^{(\ell_{ij})}_i|R^{(\ell_{ij})}_j} = 0$ by definition. We can thus bound the product
\begin{equation}
    \label{eq:RKKR_bound}
    |\braket{R_i| K_{c\not\in i}^\dagger K_c|R_j}|
    %|M_{ij}|
    \leq (1-\varepsilon_W)^{[(N-1) - 2\log N - |c| + \ell_{ij}]} \, \varepsilon_W^{(|c| + \log N -\ell_{ij})} \, \varepsilon^{(1+\log N - \ell_{ij})/2}, \text{ for $i\in g(c)$ and $j\neq i$}. 
\end{equation}
We have only shown that the bound~(\ref{eq:RKKR_bound}) holds for $j\in g(c)$. In fact, though, it also holds for $j\not\in g(c)$. The calculation of the bound in this case proceeds almost identically, with the only difference being that we also utilize the bound
\begin{equation}
     \left|\braket{R_i^{(\ell)}| K_0|W}\braket{W|K_1^\dagger K_1|R_j^{(\ell)}}\right| \leq \varepsilon_W \sqrt{\varepsilon}.
\end{equation}

We now proceed to bound the norm of the vector $M \boldsymbol{\alpha}$. We do so by first bounding the maximum absolute value column and row sum norms, $||M||_1$ and $||M||_\infty$, respectively. To bound $||M||_\infty$, observe that 
\begin{align}
    ||M||_\infty &= \max_{i} \sum_{j} |M_{ij}| = \max_{i} \sum_{j\neq i}\left| \braket{R_i| K_{c\not\in i}^\dagger  K_c|R_j}\right| \\
    & \leq 
    \sum_{j\neq i}  (1-\varepsilon_W)^{[(N-1) - 2\log N - |c| + \ell_{ij}]} \, \varepsilon_W^{(|c| + \log N -\ell_{ij})} \, \varepsilon^{(1+\log N - \ell_{ij})/2},
\end{align}
where the second line follows from \Cref{eq:RKKR_bound}. 
To perform the sum, note that there at most $2^{(\log N - \ell )}$ branches $j$ for which $\ell_{ij} = \ell$. Grouping such branches together, we have
\begin{align}
     ||M||_\infty
     &\leq \sum_{\ell = 1}^{\log N -1} 2^{(\log N - \ell )}  (1-\varepsilon_W)^{[(N-1) - 2\log N - |c| + \ell]} \, \varepsilon_W^{(|c| + \log N -\ell)} \, \varepsilon^{(1+\log N - \ell)/2} \\
     & = q(c) \sum_{\ell = 1}^{\log N -1} 
     (1-\varepsilon_W)^{[- 2\log N + \ell]} \, (2\varepsilon_W)^{( \log N -\ell)} \, \varepsilon^{(1+\log N - \ell)/2}\\
     %&= q(c) \varepsilon_W \varepsilon \frac{ (1-\varepsilon_W)^{\log N }- (2\varepsilon_W \sqrt{\varepsilon})^{\log N}}{(1-\varepsilon_W)^{2\log N} (1-\varepsilon_W-2\varepsilon_W\sqrt{\varepsilon})} \\
     & \leq 2q(c)\varepsilon^2/(1-2\varepsilon^{\frac{3}{2}})
\end{align}
where the summation does not include an $\ell = \log N$ term as a consequence of \Cref{remark:1}. The inequality on the last line is obtained by applying the bounds $(1-\varepsilon_W)\leq 1$ and $\varepsilon_W\leq \varepsilon$, evaluating the sum, then further simplifying the result by assuming that $\varepsilon\leq 1/2$.
By a nearly identical calculation, we find that $||M||_1$ obeys the same upper bound as $||M||_\infty$. Because $||O||_2^2 \leq ||O||_1 ||O||_\infty$ for arbitrary $O$, it follows that $||M||_2$ obeys this upper bound as well. Therefore,
\begin{equation}
    \label{eq:2_norm}
    ||M\boldsymbol{\alpha}||_2 \leq ||M||_2 ||\boldsymbol \alpha||_2 \leq q(c) \varepsilon^2 /(1-2\varepsilon^{\frac{3}{2}}),
\end{equation}
as $||\boldsymbol \alpha||_2 = 1$. 

We can now finally bound $|\braket{\ugood|\ubad^\parallel}|$. By the Cauchy-Schwarz inequality, 
\begin{equation}
    |\braket{\ugood|\ubad^\parallel}|
    =|\braket{\boldsymbol{\alpha_g},M\boldsymbol{\alpha}}| \leq ||\boldsymbol{\alpha_g}||_2\, ||M \boldsymbol{\alpha}||_2.
\end{equation}
Using \Cref{eq:56,eq:2_norm} together with the fact that $||\boldsymbol{\alpha_g}||_2 = \Lambda(c)$, we thus obtain
\begin{equation}
    |\braket{\ugood|\ubad^\parallel}| \leq  \varepsilon^2 q(c)\Lambda(c)/(1-2\varepsilon^{\frac{3}{2}}) \leq B \varepsilon^2  \braket{\ugood|\ugood}.
\end{equation}
where we have defined
\begin{equation}
    B\equiv \frac{1}{1-2\varepsilon^{\frac{3}{2}}}\left(\frac{1-\varepsilon_W}{1-\varepsilon}\right)^{\log N}.
\end{equation}
This inequality is equivalent to \Cref{eq:lemma2_alt}, so the proof is complete.
\end{proof}

\Cref{lem:2}, which incorporates the detrimental effects of $\ket{\ubad^\parallel}$, is analogous to \Cref{eq:good_overlap}. It remains to quantify the detrimental effects of $\ket{\ubad^\perp}$, and we do so in the following lemma.

\begin{lemma}
\label{lem:3}
The overlap of $\ket{\overline{\mathrm{bad}}(c)^{\perp}}$ with the ideal state can be bounded as
\begin{equation}
|\braket{\psi_\mathrm{out},f(c)|\overline{\mathrm{bad}}(c)^{\perp}}|
  \leq  \sqrt{p(c)}- (1-B\varepsilon^2)\sqrt{q(c)}  \, \Lambda(c) \left(\frac{1-\varepsilon}{1-\varepsilon_W}\right)^{\frac{3}{2}\log N},
\end{equation}
where $p(c)$ is the probability of error configuration $c$, defined in \Cref{eq:pc}.
\end{lemma}

\begin{proof}
The proof is a direct application of the identity
\begin{equation}
      |\braket{\psi|\overline\phi^\perp}| \leq
    \left(
    \braket{\overline\phi|\overline\phi}+
    \braket{\overline\phi^\perp|\overline\phi^\perp}
    \right)^{\frac{1}{2}}
    - |\braket{\psi|\overline\phi}|,
\end{equation}
which holds for arbitrary $\ket{\psi}$ and unnormalized, orthogonal states $\ket{\overline\phi}$ and $\ket{\overline\phi^\perp}$.  Taking $\ket{\phi} = \ket{\ugood}+\ket{\ubad^\parallel}$ and $\ket{\phi^\perp} = \ket{\ubad^\perp}$, we have
\begin{align}
    |\braket{\psi_\mathrm{out},f(c)|\overline{\mathrm{bad}}(c)^{\perp}}|
    &\leq \braket{\overline \Omega(c)|\overline \Omega(c)}^{\frac{1}{2}} - \left|\bra{\psi_\mathrm{out},f(c)}\left(\ket{\ugood}+\ket{\overline{\mathrm{bad}}(c)^{\parallel}}\right)\right| 
\end{align}
where we have used $\ket{\overline \Omega(c)} = \ket{\ugood}+\ket{\ubad^\parallel}+\ket{\ubad^\perp}$. The proof is completed by leveraging the result of \Cref{lem:2} and recognizing the first term on the right hand side as $\sqrt{p(c)}$. 

\end{proof}

\Cref{lem:2,lem:3} are the analogies of \Cref{eq:good_overlap,eq:bad_overlap} and respectively. Following the main text, we can then write down the analogy of \Cref{eq:config_fid_bound} using the reverse triangle inequality
\begin{equation}
    \label{eq:83}
    \overline F(c) \geq 
    \begin{cases} 
    \left(2(1-C)\Lambda(c)\sqrt{q(c)} - \sqrt{p(c)} \right)^2, &\mbox{for }  2(1-C)\Lambda(c)\sqrt{q(c)} \geq \sqrt{p(c)}, \\
    0, &\mbox{otherwise. }
    \end{cases}
\end{equation}
where, to simplify notation, we have defined
\begin{equation}
    1-C\equiv (1-B\varepsilon^2)
    \left(\frac{1-\varepsilon}{1-\varepsilon_W}\right)^{\frac{3}{2}\log N}.
\end{equation}
Now, recall that the query fidelity is given by $F = \sum_c \overline{F}(c)$. We can equivalently express the query fidelity as an expectation value with respect to the distribution $q(c)$ as
\begin{equation}
    F = \mathbb{E}\left(\overline{F}/q\right) \equiv \sum_c q(c) \left(\frac{\overline{F}(c)}{q(c)} \right).
\end{equation}
The utility of expressing $F$ this way will be made apparent shortly.
Analogously to \Cref{eq:infid_bound_intermediate0,eq:infid_bound_intermediate} in the main text, we can then bound the query fidelity as
\begin{align}
    F &\geq \mathbb{E}\left(\sqrt{\overline{F}/q}\right)^2 = \left[\sum_c \sqrt{q(c) \overline{F}(c)} \right]^2
    \\
    &\geq \left[2 (1-C) \mathbb{E}(\Lambda) - \mathbb{E}(\sqrt{p/q})\right]^2
\end{align}
where the second line follows from \Cref{eq:83} under the assumption that $2 (1-C) \mathbb{E}(\Lambda) \geq \mathbb{E}(\sqrt{p/q})$. Let us now consider the two expectation values. We have
\begin{equation}
    \mathbb{E}(\sqrt{p/q}) = \sum_c \sqrt{q(c)p(c)} \leq 1,
\end{equation}
where the inequality follows from recognizing the sum as the Bhattacharyya distance between the distributions $p$ and $q$. 
We also have
\begin{equation}
    \mathbb{E}(\Lambda) = \sum_c q(c) \Lambda(c) = (1-\varepsilon_W)^{\log N},
\end{equation}
which follows from a recursive calculation nearly identical to the one described in the main text. Thus, we find
\begin{equation}
    F \geq \left[2 (1-C) (1-\varepsilon_W)^{\log N} - 1\right]^2
\end{equation}
in analogy to \Cref{eq:infid_bound_intermediate2}. Employing Bernoulli's inequality, we obtain
\begin{align}
    F &\geq \left[ 2(1-C)(1-\varepsilon_W \log N) - 1 \right]^2 \\
    & \geq 1 - 4\varepsilon_W \log N - 4C(1-\varepsilon_W \log N),
\end{align}
assuming $\varepsilon_W \log N + 4C(1-\varepsilon_W\log N)\leq 1/4$.
Equivalently,
\begin{align}
    1-F &\leq  4\varepsilon_W \log N + 4C(1-\varepsilon_W \log N).
\end{align}

Let us simplify this bound on the infidelity so that the scaling with $N$ is evident. We use Bernoulli's inequality again to obtain the bound
\begin{align}
    C &= 1 - (1-B\varepsilon^2)\left(\frac{1-\varepsilon}{1-\varepsilon_W}\right)^{\frac{3}{2}\log N} \\
    & \leq \frac{3}{2}\log N \frac{\varepsilon-\varepsilon_W}{1-\varepsilon_W}(1+B\varepsilon^2).
\end{align}
Inserting this bound, and using the definition of $B$, we find
\begin{equation}
    \label{eq:finally}
    1-F \leq A \varepsilon \log N
\end{equation}
where we have defined
\begin{align}
    \label{eq:finallyA}
    A &\equiv 
    4\frac{\varepsilon_W}{\varepsilon} + 6\frac{1-\varepsilon_W/\varepsilon}{1-\varepsilon_W}\left(1+\frac{\varepsilon^2}{1-2\varepsilon^{\frac{3}{2}}}\left(\frac{1-\varepsilon_W}{1-\varepsilon}\right)^{\log N}\right) (1-\varepsilon_W \log N)
\end{align}
In the relevant limit of $\varepsilon \log N \ll 1$, the coefficient $A$ is well-approximated by keeping only the leading order terms,
\begin{equation}
    A \approx 6 - 2\frac{\varepsilon_W}{\varepsilon}.
    %4\frac{\varepsilon_W}{\varepsilon} + 6(1 - \varepsilon_W/\varepsilon)
\end{equation}
Notice that, for mixed-unitary error channels, for which $\varepsilon_W = \varepsilon$, we have $A\approx 4$, in agreement with the proof in the main text.

This concludes the proof of the favorable error scaling, \Cref{eq:finally}. Recall, however, that we have made two simplifying assumptions. As described in \Cref{sec:proof_error_model}, we have assumed that errors are only applied at one time step and that the error channel has a Kraus rank of 2. It remains to relax these assumptions, and we do so in the following subsections. 

\subsection*{Error channels with Kraus rank \texorpdfstring{$\neq 2$}{TEXT}}
The proof can be straightforwardly adapted to handle the case of channels with Kraus rank $\neq 2$. First, let us consider the case of Kraus rank $=1$, for which the error channel contains a single unitary Kraus operator $K_0$. This case can be understood as describing coherent errors in the routers and is actually already covered by the above proof. One simply sets $K_1 = 0$, and the bound \Cref{eq:finally} holds.

Next we consider the Kraus rank $> 2$ case. We define two error configurations, $c$ and $c'$, to be \emph{similar} if and only if $K_{c(r,t)} = K_0 \iff K_{c'(r,t)} = K_0$. That is, error configurations are similar if errors occur at the same locations, though which error $K_{m>0}$ occurs at a given location can differ. We further define $c$'s similar set as the set of all configurations $c'$ that are similar to $c$.

Let $|c_m|$ denote the number of errors $K_{m>0}$ in configuration $c$. We generalize the definition of $q(c)$ as
\begin{equation}
    q(c) = (\braket{W|K_0^\dagger K_0|W})^{(N-1) - |c|} \prod_{i>0} \left(\braket{W|K_m^\dagger K_m|W}\right)^{|c_m|},
\end{equation}
where $|c| = \sum_m|c_m|$. Recall that $\varepsilon_W \equiv 1- \braket{W|K_0^\dagger K_0|W}$. Then,
because
\begin{equation}
    \sum_{i>0} \braket{W|K_i^\dagger K_i|W} = 1- \varepsilon_W,
\end{equation}
as a consequence of the Kraus operators' completeness relation,
the total probability of obtaining any configuration from $c$'s similarity set is
\begin{equation}
    (1-\varepsilon_W)^{(N-1)-|c|}\varepsilon_W^{|c|}.
\end{equation} 
This quantity matches the definition of $q(c)$ for the Kraus rank $=2$ case. 
Therefore, by grouping all error configurations in the same similarity set together, the proof proceeds exactly as above. In other words, the proof is agnostic to the kind of errors that occur; so long as the total probability of \emph{any} error occurring is fixed, the result is the same. 

\subsection*{Errors at all time steps}
The core idea of the proof still holds in the case when errors occur at all time steps. In particular, the arguments above can be iterated at each time step in order to isolate the component of the final state that lies along the ideal state. 
Accordingly, we define 
\begin{equation}
    \ket{\ugood} = \sum_{i\in g(c)} \alpha_i \tilde{a}_i \ket{i}^A \ket{x_i}^B \ket{\overline f_i(c)}
\end{equation}
which differs from the definition of $\ket{\ugood}$ given above only in the coefficients $\tilde a_i$. These coefficients are associated with the repeated application of \Cref{eq:K0_not_iden}, which we now apply $T$ times to each of the $\log N$ routers in branch $i$. We extend the definition \Cref{eq:a_psi} to
\begin{align}
    \Re\left[\prod_{i=1}^{m} a_{\psi_i}\right] &\geq (1-\varepsilon)^{m/2}, \text{ for any $m\leq T\log N$},
\end{align}
from which it follows that $\Re\tilde a_i \geq (1-\varepsilon)^{\frac{1}{2}T\log N}$. \Cref{lem:good} is then generalized simply by replacing $\log N \rightarrow T \log N$.
Conceptually, the idea is that each time the no-error operator $K_0$ is applied to an active router, we pay the price of a $\sim (1-\varepsilon)^{\frac{1}{2}}$ prefactor in order to isolate the ideal component resultant state. When we repeat this procedure $T$ times for each of the $\log N$ routers, these prefactors multiply, giving a total prefactor of $\sim(1-\varepsilon)^{\frac{1}{2}T\log N}$. 

Next consider the generalization of \Cref{lem:2}, which asserts that $|\braket{\ugood|\ubad}|/\braket{\ugood|\ugood}= O(\varepsilon^2)$. Observe the proof of \Cref{lem:2} does not actually depend on the time step when the errors are applied. For example, at time $t<t^*$, the QRAM has active routers in only the first $< \log N$ levels of the tree. One can equivalently view this state as the state of a QRAM with $< \log N$ levels at time step $t^*$, and so the proof of \Cref{lem:2} directly applies.
When errors occur at multiple time steps,
it follows that
\begin{equation}
    \label{eq:lem2again}
    |\braket{\ugood|\ubad}| \leq B'\varepsilon^2 \braket{\ugood|\ugood},
\end{equation}
for some $B'$.
In general $B'\neq B$, though $B'\approx 1$ in the limit $\varepsilon \log N \ll 1$.
An equivalent statement of this result is that, regardless of when errors occur, the overlap of the good and bad states always contains sufficiently many off-diagonal matrix elements (like $\braket{\psi|K_0^\dagger K_0|\psi^\perp}\sim\varepsilon$) such that $\braket{\ugood|\ubad}\sim \varepsilon^2$ to leading order.

Having generalized \Cref{lem:good,lem:2}, the generalization of \Cref{lem:3} follows straightforwardly. One simply replaces $\log N \rightarrow T \log N$ and $B\rightarrow B'$ in the bound.

With these generalizations in hand, one can use the same argument as above to show that the fidelity can be expressed in terms of $\mathbb{E}(\Lambda)$, where the expectation is now taken with respect to error configurations with errors at multiple time steps. $\mathbb{E}(\Lambda)$ was already computed with respect to such configurations in the main text. Applying the same calculation, we obtain
\begin{equation}
    \mathbb{E}(\Lambda) = (1-\varepsilon_W)^{T\log N}.
\end{equation}
Proceeding as above, we then find
\begin{equation}
    1-F \leq A' \varepsilon T \log N
\end{equation}
where, in the limit $\varepsilon \log N \ll 1$, the coefficient $A'$ can be well-approximated by
\begin{equation}
    A' \approx 6 - 2\frac{\varepsilon_W}{\varepsilon}.
\end{equation}
In general, $A'\neq A$ since $B\neq B'$.

\end{widetext}

\end{document}